\documentclass[12pt]{article}
\usepackage{amsmath}
\usepackage{amsfonts}
\usepackage{amsthm}
\usepackage{latexsym}
\usepackage{amscd}
\usepackage{amssymb}
\usepackage{graphicx}
\usepackage{float}
\usepackage{scalerel}
\usepackage[applemac]{inputenc}
\usepackage{dsfont}
\usepackage{multirow}
\usepackage{dcolumn}
\usepackage{rccol}
\usepackage{enumitem}
 \usepackage[text={6.6in,8.6in},centering]{geometry} 
\usepackage{scalefnt}
\usepackage{booktabs}
\usepackage{subfigure}
\usepackage{caption}
\usepackage{color}
\usepackage[normalem]{ulem}
  \usepackage[overload]{empheq}
\usepackage{color}
\usepackage[bookmarks=false,colorlinks]{hyperref}
\usepackage{cleveref}
\usepackage{textcomp}
\usepackage{array}
\usepackage{supertabular}
\usepackage{hhline}
\usepackage{supertabular}
\usepackage{booktabs}
\usepackage{cancel}
\usepackage{multirow}
\usepackage{mathtools}
\def\block(#1,#2)#3{\multicolumn{#2}{c}{\multirow{#1}{*}{$ #3 $}}}

\makeatletter
\newcommand\arraybslash{\let\\\@arraycr}
\makeatother
\newcolumntype{+}{>{\global\let\currentrowstyle\relax}}
\newcolumntype{^}{>{\currentrowstyle}}
      
\usepackage{xr}
\newcommand\dsum{\displaystyle\sum}

\newlength{\bracewidth}

\def\fudge{\mathchoice{}{}{\mkern.5mu}{\mkern.8mu}}
\def\bbc#1#2{{\rm \mkern#2mu\vbar\mkern-#2mu#1}}
\def\bbb#1{{\rm I\mkern-3.5mu #1}}
\def\bba#1#2{{\rm #1\mkern-#2mu\fudge #1}}
\def\bb#1{{\count4=`#1 \advance\count4by-64 \ifcase\count4\or\bba A{11.5}\or
   \bbb B\or\bbc C{5}\or\bbb D\or\bbb E\or\bbb F \or\bbc G{5}\or\bbb H\or
   \bbb I\or\bbc J{3}\or\bbb K\or\bbb L \or\bbb M\or\bbb N\or\bbc O{5} \or
   \bbb P\or\bbc Q{5}\or\bbb R\or\bbc S{4.2}\or\bba T{10.5}\or\bbc U{5}\or
   \bba V{12}\or\bba W{16.5}\or\bba X{11}\or\bba Y{11.7}\or\bba Z{7.5}\fi}}
\def \R {\bb R}

\def \diag {\textrm{diag}}
\restylefloat{float}

\newtheorem{theorem}{Theorem}[section]

\newtheorem{remark}{Remark}[section]
\newtheorem{lemma}{Lemma}[section]

\title{Global Analysis of Multi-Host and Multi-Vector Epidemic Models}

\author{Derdei Bichara\\ Department of Mathematics \& Center for Computational and
 Applied \\Mathematics, California State University, Fullerton, CA 92831, USA.
}

\date{}

\begin{document}
\maketitle

\begin{abstract}
We formulate a multi-group and multi-vector epidemic model in which hosts' dynamics is captured by staged-progression $SEIR$ framework and the dynamics of vectors is captured by an $SI$ framework. The proposed model describes the evolution of a class of zoonotic infections where the pathogen is shared by $m$ host species and transmitted by $p$ arthropod vector species. In each host, the infectious period is structured into $n$ stages with a corresponding infectiousness parameter to each vector species. We determine the basic reproduction number $\mathcal{R}_0^2(m,n,p)$ and investigate the dynamics of the systems when this threshold is less or greater than one. We show that the dynamics of the multi-host, multi-stage, and multi-vector system is completely determined by the basic reproduction number and the structure of the host-vector network configuration. Particularly, we prove that the disease-free \mbox{equilibrium} is globally asymptotically stable (GAS) whenever $\mathcal{R}_0^2(m,n,p)<1$, and a unique strongly endemic equilibrium exists and is GAS if $\mathcal{R}_0^2(m,n,p)>1$ and the host-vector configuration is irreducible. That is, either the disease dies out or persists in all hosts and all vector species.
\end{abstract}



\section*{Introduction}
Nearly two-thirds of all known human infectious diseases (ID) are caused by zoonotic pathogens which are transmissible from one host species (humans and vertebrate animals) to another \cite{Kareshetal2012,taylor2001risk}, and therefore multi-host. Moreover, 75\% of emerging and re-emerging infectious diseases are classified as zoonoses and constitute a major public health problem around the world, responsible for over one million death and hundreds of millions of cases each year \cite{WorldBankZoonose}. Furthermore, it is estimated that zoonoses cause over 20 billion and 200 billion USD of direct and indirect economic burden across the world respectively \cite{Kareshetal2012,WorldBankZoonose}. Although the morbidity and mortality of most ID decreased, the incidence of zoonoses have increased \cite{kilpatrick2012drivers}. Therefore, understanding the dynamics of zoonoses by systematic modeling and analyzing in order to control and mitigate these scourges should be a worldwide priority.

Multi-host infectious diseases include  Lyme disease, tick-borne relapsing fever (TBRF), West Nile virus (WNV), Chagas disease, type A influenzas, Rift Valley fever, severe acute respiratory syndrome (SARS), etc.  However, 40\% of multi-host ID are vector-borne \cite{johnson2015spillover}. 
 That is, blood-sucking arthropod vectors such as ticks, mosquitoes, fleas and sandflies play the role of connecting multiple hosts while potentially infecting them and getting infected during this process. 
 Thus, formulating the dynamics of zoonoses requires understanding the ecology of all involved host species and their interactions with vector species that carry the pathogen between hosts.  
 %
 These complexities have made a comprehensive study of zoonoses very challenging, making mathematical models of zoonoses scarce and mainly focused on WNV \cite{BowGum05} or directly transmitted zoonoses \cite{begon1994host,begon1992disease,lloyd2009epidemic}. Recent research has been focused on understanding the dynamics and control of vector-borne zoonoses with one vector and two hosts (see \cite{robertson2016host,simpson2011vector} for WNV, \cite{cruz2012control} for Chagas' disease, and \cite{johnson2016modeling,palmer2018dynamics} for TBRF). Also, authors in \cite{BicharaIggidrSmith2017} proposed a class of vector-borne zoonoses with an arbitrary number of hosts and one vector and provided the complete global dynamics of equilibria.

However, another layer of complexity regarding the ecology of zoonoses consists of sometimes different arthropod vector species, or the same arthropod vectors but different genera, are responsible of transmiting the same pathogen to a number of different hosts. For instance, over 65 different mosquito species transmit WNV to a number of hosts including humans, mammals and many species of birds \cite{hamer2008culex,ReisenVectorWNV2006}. Another example of a highly complex zoonosis is the eastern equine encephalitis virus (EEEV). Indeed, as illustrated in \cite{CDCEEEV,magori2013population}, the main vector of EEEV is the mosquito \textit{Culiseta melanura}, which feeds exclusively on birds, and thus infects birds only. These birds infect, in turn, other mosquito species, which then bite humans, therefore creating a direct chain of transmission of the pathogen between a host and vector species with no direct link and/or transmissibility. Hence, as stated in \cite{magori2013population}, the elimination or mitigation of zoonoses requires breaking the multiple transmission cycles corresponding to each potential host and vector species.

Therefore, a better understanding of zoonoses requires taking into account the ecology of all host species (including \textit{dead-end} hosts -- hosts that do not contribute to further transmission of the pathogen) and all vector species, along with the epidemiology of the zoonosis within the before-mentioned species. Moreover, developing general theories for the role of intermediate hosts in pathogen emergence is one of the nine challenges in modeling the emergence of novel pathogens according to \cite{lloyd2015nine}. The goal of this paper is twofold.

The first goal of this paper is to derive a class of models that describes the interaction between $m$ host species and $p$ vector species where the latter differ in their propensity to acquire or infect the pathogen from the former. Moreover, we structure each host's infectious stage into $n$ different ``ages" or classes where they infect each vector species at different rates. 
 The proposed model provides a general framework in modeling zoonoses as it takes into account the complex patterns and multi-faceted host species dynamics along with their interactions with vector species. The derived class of models handles multiple levels of organizations including the case where some host or vector species have different epidemiological structures than others.  Our modeling framework offers a plethora of possible scenarios and could be applied to study specific cases of zoonoses, and thus hopes to provide a forum that could help guide decisions on control efforts to mitigate and/or eradicate some zoonoses. 

The second goal of the paper is to study the global asymptotic behavior of the proposed system. Particularly, we derive conditions under which the disease dies out or persists in all host species and vector species. A key threshold quantity happens to be the basic reproduction number $\mathcal R_0^2(m,n,p)$ for the general system with $m$ host, $n$ stages, and $p$ vector species. We prove that the disease-free equilibrium is globally asymptotically stable (GAS) whenever  $\mathcal R_0^2(m,n,p)<1$. 
 The proof of the uniqueness of an endemic equilibrium (EE) for large epidemic systems is known to be challenging. Indeed, the uniqueness of an EE may not hold for multi-group directly transmitted diseases \cite{MR1154787,JacSimKoo91} and sexually transmitted diseases \cite{HuangCookeCCC92}. For our system, we will prove indeed that it has a unique \textit{strongly} EE (in the sense of Thieme \cite{MR1993355}) under the assumption that the host-vector network configuration is irreducible and $\mathcal R_0^2(m,n,p)>1$. To do so, we  transformed the system at equilibrium into an auxiliary dynamical system and showed that the equilibrium of the newly crafted system, is globally attractive under the pre-cited hypotheses. Moreover, we will prove that the \textit{strongly} EE is GAS whenever it exists. The latter relies on a carefully constructed Lyapunov function and elements of graph theory. 
The authors in \cite{Guo_li_CAMQ_06,Guo_li_PAMS08} first used tools of graph theory to study the GAS of the EE for multi-group models for directed transmitted diseases. To do so, they derived a Lyapunov function for the multi-group system from that of a single epidemic system. Also, in \cite{shuai2011global,shuai2013global}, the authors used these tools to investigate for stage-progression models for water-borne diseases and  the authors in \cite{iggidr2014dynamics} successfully generalized the approach to a multi-group vector-borne $SIR-SI$ system. For our model, the arbitrary number of host species, stages, and vector species (all potentially different) present specificities that make the approach employed the after-mentioned papers impractical. For instance, the construction of the Lyapunov function from that of one group does not apply in our case.  
  
The paper is outlined as follows: \Cref{GeneralModelFormulation} is devoted to the derivation of the model. 
 \Cref{BasicProperties} lays out the basic the properties of the model and determines the associated basic reproduction number. 
\Cref{GASEquilibria} provides a complete analysis of the model,  and \Cref{ConclusionDiscussions} is devoted to concluding remarks and discussions.
\section{Formulation of the Model}\label{GeneralModelFormulation}
We consider the dynamics of a disease transmitted by the interplay between $m$ host species and $p$ arthropod vectors. We assume that the disease follow an $SEI^nR-SI$ structure where $n$ designates the number of infectious stages in the evolution of the disease within each host species. Let $N_i$ represents the total population of each host species $i$ ($1\leq i\leq m$) and let $S_i$, $E_i$, $I_{l,i}$, and $R_{i}$ denote the susceptible, latent, infectious at stage $l$ ($1\leq l\leq n$), and recovered populations of host species $i$, respectively. Let $N_{v,j}$ be the total population of arthropod vectors of species $j$, ($1\leq j\leq p$), each of which is composed by susceptible vectors $S_{v,j}$ and infectious vectors $I_{v,j}$.
The susceptible populations of host species $i$ are generated through a constant recruitment $\Lambda_i$,  subjected to a natural death rate of $\mu_i$ and could be infected after being bitten by an infectious vector of any species $j$. 
After being infected, the susceptible populations of Host $i$ become latent, who then become infectious after an incubation period of $1/\nu_i$. The infectiosity period of host species is structured into $n$ stages, each characterizing the level of parasitemia in the corresponding host. At each stage $l$, the infectious of host $i$ recover at a rate $\eta_{l,i}$, progressing to the next stage of infection at a rate $\gamma_{l,i}$ or naturally die at a rate $\mu_i$. 
The susceptible arthropod vectors of species $j$, $S_{v,j}$, could be infected by all infectious hosts, of any species and of any stage, at a different rates. Moreover, they are replenished at a constant rate $\Lambda_{v,j}$ and die either by natural death, at rate $\mu_{v,j}$, or due to control strategies specific for each vector of species $j$, at a rate $\delta_{v,j}$. 

 The overall multi-host, multi-stage and multi-vector model is captured by the following system:
\begin{equation} \label{MultiHostMultiVector}
\left\{\begin{array}{llll}
\displaystyle\dot S_{i}=\Lambda_i-\sum_{j=1}^pa_{i,j}\beta_{i,j}^\diamond S_{i}\dfrac{I_{v,j}}{N_{i}}-\mu_i S_i\\
\dot E_{i}=\dsum_{j=1}^pa_{i,j}\beta_{i,j}^\diamond S_{i}\dfrac{I_{v,j}}{N_{i}}-(\mu_i+\nu_i)E_{h}\\
\dot I_{1,i}=\nu_iE_i-(\mu_i+\eta_{1,i}+\gamma_{1,i}) I_{1,i}\\
\dot I_{2,i}=\gamma_{1,i}I_{1,i}-(\mu_i+\eta_{2,i}+\gamma_{2,i}) I_{2,i}\\
\vdots\\
\dot I_{n-1,i}=\gamma_{n-2,i}I_{n-2,i}-(\mu_{i}+\eta_{n-1,i}+\gamma_{n-1,i}) I_{n-1,i}\\
\dot I_{n,i}=\gamma_{i,n-1}I_{i,n-1}-(\mu_{i}+\eta_{n,i}) I_{n,i}\\
 \dot R_i=\sum_{k=1}^n\eta_{k,i}I_{i}-\mu_iR_i\\
\displaystyle\dot S_{v,j}=\Lambda_{v,j}-\sum_{i=1}^ma_{i,j} S_{v,j}\sum_{l=1}^n\dfrac{\beta_{l,j}^i I_{l,i}}{N_{i}}-(\mu_{v,j}+\delta_{v,j})S_{v,j}\\
\displaystyle\dot I_{v,j}=\sum_{i=1}^ma_{i,j} S_{v,j}\sum_{l=1}^n\dfrac{\beta_{l,j}^i I_{l,i}}{N_{i}}-(\mu_{v,j}+\delta_{v,j})I_{v,j}.
\end{array}\right.
\end{equation}
A schematic description of the model is captured by \Cref{fig:FlowSEIRMultiHostMultiVector} and the parameters are described in \Cref{TableMHMSMV}.
\begin{table}[h!]\label{TableMHMSMV}
    \caption{Description of the parameters used in System (\ref{MultiHostMultiVectorCompact}).}
    \label{tab:Param}
        \setlength{\tabcolsep}{0em}
    \begin{tabular}{llllllllllll}
     \toprule
      Parameters & Description \\
      \midrule
 $\mathbf{\Lambda_h} =[\Lambda_1,\Lambda_2,\dots,\Lambda_m]^T$  & Vector of recruitment of the hosts;\\
 $\mathbf{\Lambda_v} =[\Lambda_{v,1},\Lambda_{v,2},\dots,\Lambda_{v,p}]^T$  & Vector of recruitment of the vectors;\\
 $a_{i,j}$  & Biting rate of vector $j$ on Host $i$; \\
$ \beta_{i,j}^\diamond$  & Infectiousness of Host $i$ to vectors $j$ per biting; \\
$\mu_h =[\mu_1,\mu_2,\dots,\mu_m]^T$  & Hosts' death rate;\\
$\nu_h =[\nu_1,\nu_2,\dots,\nu_m]^T$  & Hosts' incubation rate;\\
$\alpha_k =[\alpha_{k,1},\alpha_{k,2},\dots,\alpha_{k,m}]^T$  &  Hosts' total duration at infectious stage $k$;\\
$\eta_k =[\eta_{k,1},\eta_{k,2},\dots,\eta_{k,m}]^T$  & Hosts' recovery rate at stage $k$;\\
$\gamma_k =[\gamma_{k,1},\gamma_{k,2},\dots,\gamma_{k,m}]^T$  & Hosts' progression rate from infectious stage $k$ to {\small{$k+1$}};\\
$\beta_{k,j}^i$  & Vector of species $j$'s infectiousness to Host $i$ at stage {\small{$k$}};\\ 
$\mu_v=[\mu_1,\mu_2,\dots,\mu_p]^T$  & Vectors' natural mortality rates;\\
$\delta_v=[\delta_1,\delta_2,\dots,\delta_p]^T$  & Vectors' control-induced mortality rates.\\
  \bottomrule
    \end{tabular}
\end{table}
 The total population of host species and vector species are asymptotically constant as their dynamics are given by $\dot N_i=\Lambda_i-\mu_iN_i$ and $\dot N_{v,j}=\Lambda_{v,j}-(\mu_{v,j}+\delta_{v,j})N_{v,j}$, respectively. Hence, by using the theory of asymptotic systems \cite{CasThi95,0478.93044} and by denoting the limits of $N_i$ and $N_{v,j}$ by $N_i=\frac{\Lambda_i}{\mu_i}$ and $N_{v,j}=\frac{\Lambda_{v,j}}{\mu_{v,j}+\delta_{v,j}}$, System (\ref{MultiHostMultiVector}) is asymptotically equivalent to its limit system. Moreover, in this case, it could be written in a compact form as:
\begin{equation} \label{MultiHostMultiVectorCompact}
\left\{\begin{array}{llll}
\displaystyle\dot{\mathbf{S}}= \mathbf{\Lambda}_{h} -\diag^{-1}(\mathbf{N}_h)\diag(\mathbf{S})A\circ B^\diamond \mathbf{I}_v-\diag(\mu_h) \mathbf{S}\\
\dot{\mathbf{E}}=\diag^{-1}(\mathbf{N}_h)\diag(\mathbf{S})A\circ B^\diamond \mathbf{I}_v-\diag(\mu_h+\nu_h)\mathbf{E}\\
\dot{\mathbf{I}}_1=\diag(\nu)\mathbf{E}-\diag(\alpha_1)\mathbf{I}_1\\
\dot{\mathbf{I}}_2=\diag(\gamma_{1})\mathbf{I}_1-\diag(\alpha_2)\mathbf{I}_2\\
\vdots\\
\dot{\mathbf{I}}_{n-1}=\diag(\gamma_{n-2})\mathbf{I}_{n-2}-\diag(\alpha_{n-1})\mathbf{I}_{n-1}\\
\dot{\mathbf{I}}_{n}=\diag(\gamma_{n-1})\mathbf{I}_{n-1}-\diag(\alpha_{n})\mathbf{I}_{n}\\
\displaystyle\dot{\mathbf{S}}_v= \mathbf{\Lambda}_{v} -\diag(\mathbf{S}_v)\sum_{l=1}^n(A\circ B_l)^T\diag^{-1}(\mathbf{N}_h)\mathbf{I}_l- \diag(\mu_v+\delta_v)\mathbf{S}_v\\
\displaystyle\dot{\mathbf{I}}_v=\diag(\mathbf{S}_v)\sum_{l=1}^n(A\circ B_l)^T\diag^{-1}(\mathbf{N}_h)\mathbf{I}_l- \diag(\mu_v+\delta_v)\mathbf{I}_v,
\end{array}\right.
\end{equation}

where  $\mathbf{S} =[S_1,S_2,\dots,S_m]^T$, $\mathbf{N}_h =[N_1,N_2,\dots,N_m]^T$, $\mathbf{E} =[E_1,E_2,\dots,E_m]^T$,
 $\mathbf{I}_k =[I_{k,1},I_{k,2},\dots,I_{k,m}]^T$ is the vector of infectious at stage $k$ ($1\leq k\leq n$) for all hosts species. Also, $\alpha_{k,i}=\mu_{k,i}+\eta_{k,i}+\gamma_{k,i}$ and $\gamma_{n,i}=0$ for all $i$. 
  For vectors, $\mathbf{S}_v =[S_{v,1},S_{v,2},\dots,S_{v,p}]^T$ and $\mathbf{I}_v =[I_{v,1},I_{v,2},\dots,I_{v,p}]^T$ denote the vectors of susceptible and infected, respectively. The matrices $A$, $B^\diamond$ and $B_l$ are given by
 $$A=\left(\begin{array}{cccc}
a_{1,1} & a_{1,2} & \dots & a_{1,p}\\
a_{2,1} & a_{2,2} & \dots & a_{2,p}\\
\vdots &\vdots&\ddots&\vdots\\
a_{m,1} & a_{m,2} & \dots & a_{m,p}\\ 
\end{array}\right),\;
B^\diamond=({\beta_{i,j}^\diamond})=\left(\begin{array}{cccc}
\beta_{1,1}^\diamond & \beta_{1,2}^\diamond & \dots & \beta_{1,p}^\diamond\\
\beta_{2,1}^\diamond & \beta_{2,2} ^\diamond& \dots & \beta_{2,p}^\diamond\\
\vdots &\vdots&\ddots&\vdots\\
\beta_{m,1}^\diamond & \beta_{m,2}^\diamond & \dots & \beta_{m,p}^\diamond\\ 
\end{array}\right),$$
$$B_k=\left(\begin{array}{cccc}
\beta_{k,1}^1 & \beta_{k,2}^1 & \dots & \beta_{k,p}^1\\
\beta_{k,1}^2 & \beta_{k,2}^2 & \dots & \beta_{k,p}^2\\
\vdots &\vdots&\ddots&\vdots\\
\beta_{k,1}^m & \beta_{k,2}^m & \dots & \beta_{k,p}^m\\ 
\end{array}\right),$$
\begin{figure}[ht]
\centering
\includegraphics[scale =0.65]{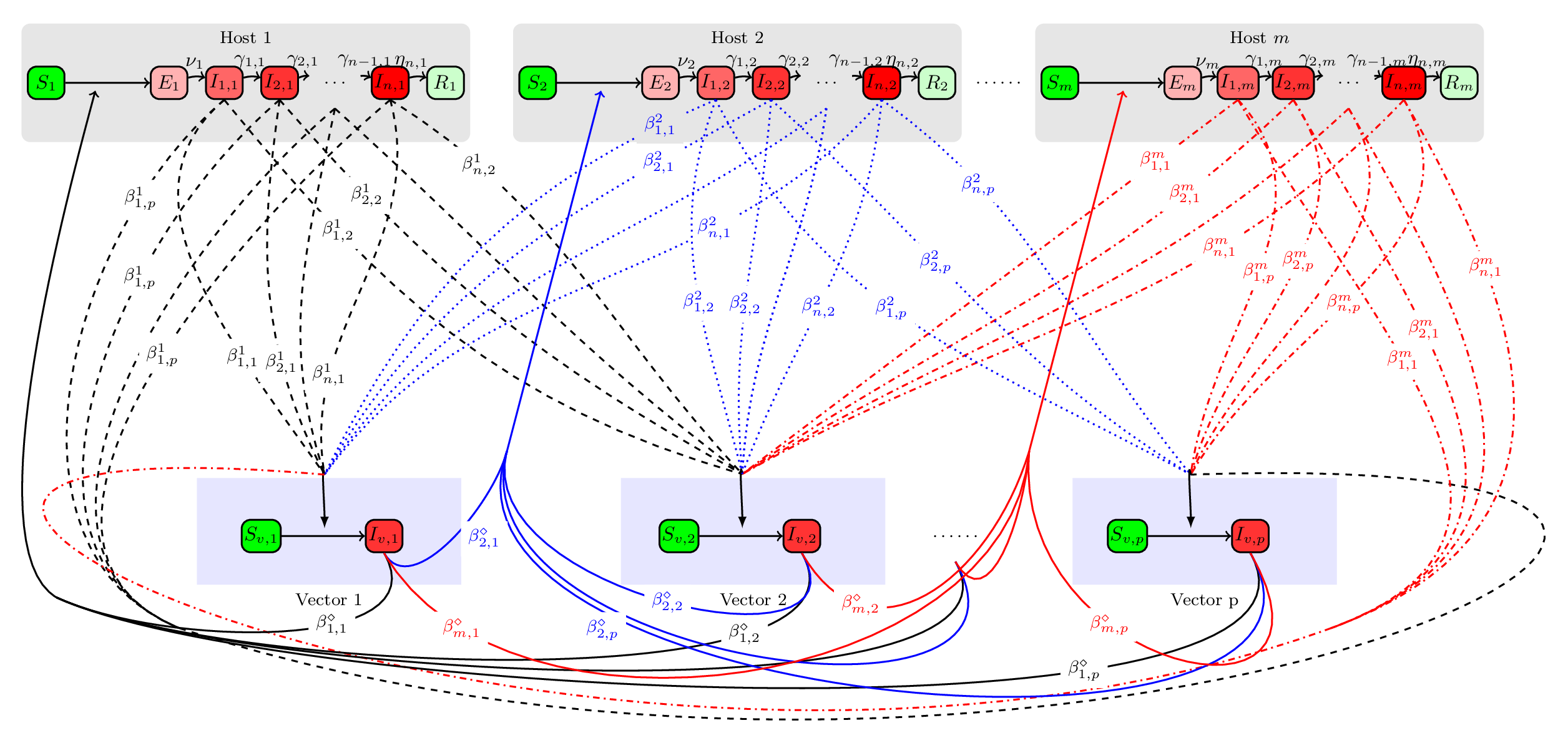}
\caption{Flow diagram of Model \ref{MultiHostMultiVector}. The dashed, dotted and dash-dotted lines capture the infection routes of vectors from all the infectious hosts at different stages. The planes lines captures the infection routes of hosts from the vectors. We did not display the recovery routes from each $I_{k,i}$ (for $k=1,2\dots,n-1$ and $i=1,2,\dots,m$) to $R_i$. }
\label{fig:FlowSEIRMultiHostMultiVector} 
\end{figure}
for $1\leq k\leq n$. $A$ represents the biting/landing rates matrix, $B^\diamond$ captures the infectiousness of host species to vector species; and the matrices $B_k$ represents the infectiousness of vector species from infected host species at stage $k$.

Model (\ref{MultiHostMultiVectorCompact}) describes the dynamics of a Multi-Host and Multi-Vector $SEI^nR-SI$ model where the infectious stage in each host is composed of $n$ stages. The model is flexible and could be adapted to cases where different hosts could have different epidemiological structures with respect to the infection. This extends the work in \cite{BicharaIggidrSmith2017,johnson2016modeling,palmer2018dynamics} which consider the dynamics zoonoses and one vector species. Our model extend also the multi-host and multi-vector in \cite{cruz2012multi} which explores and SIR-SI type of host-vector interaction. Particularly, our model consists of $SEI^nR$ where the infection of each host is stratified into $n$ infectious classes and each of these classes infect susceptible vectors at different rates. This extension captures a more realistic aspect of infection in the interactions between hosts and vectors. Moreover, although we considered an $SI$ structure for the vectors for simplicity, the results of this paper are also valid if the vector species's dynamics follows an $SEI$ structure. Indeed, for some vectors such as mosquitoes, their incubation period is often nearly two weeks, which is on the same scale as their lifespan, making an $SEI$ model more suited for their dynamics.

\section{Basic properties and reproduction number}\label{BasicProperties}

 The set 
%
$$\Omega=\left \{ (\mathbf{S}, \mathbf{E}, \mathbf{I}_1,\dots, \mathbf{I}_n,\mathbf{S}_v,\mathbf{I}_v)\in\R^{m(n+2)+2p}_+\mbox{\large $\mid$} 
\mathbf{S}+\mathbf{E}+\dsum_{i=1}^n\mathbf{I}_i\leq\mathbf{\Lambda_h}\circ\dfrac{1}{\mu}, \;\mathbf{S}_v+\mathbf{I}_v\leq\mathbf{\Lambda_v}\circ\dfrac{1}{\mu_v+\delta_v} \right \} $$
%
where $\circ$ denotes the Hadamard product, is a compact attracting  positively invariant for System (\ref{MultiHostMultiVectorCompact}). Therefore, the solutions of System (\ref{MultiHostMultiVectorCompact}) are biologically  substantiated. The trivial equilibrium of System (\ref{MultiHostMultiVectorCompact}) is the \textit{disease-free equilibrium} (D.F.E) and is given by \mbox{$E_0=\left(\bar{\mathbf{S}}, \mathbf{0}_{m(n+1)},\bar{\mathbf{S}}_v,\mathbf{0}_p
\right)$} where $$\bar{\mathbf{S}}=\mathbf{\Lambda_h}\circ\dfrac{1}{\mu}\quad\textrm{and}\quad \bar{\mathbf{S}}_v=\mathbf{\Lambda_v}\circ\dfrac{1}{\mu_v+\delta_v}.$$

\noindent In the following lemma, we derive the basic reproduction number for Model (\ref{MultiHostMultiVectorCompact}) following \cite{MR1057044,VddWat02}. 
\begin{lemma}\label{R0mnp}
The basic reproduction number of Model (\ref{MultiHostMultiVectorCompact}) is given by:
\begin{multline*}
 \mathcal{R}_0^2(m,n,p)=\rho\left(\frac{}{}\diag(\mathbf{N}_v)\left(\dsum_{l=1}^n(A\circ B_l)^T \diag(\gamma_1\circ\gamma_2\circ\dots\circ\gamma_{l-1}\circ\nu)\right.\right.\\
 \left.\left.\diag^{-1}(( \mu+\nu)\circ\alpha_1\circ\alpha_2\circ\dots\circ\alpha_{l})\frac{}{}\right)
 \diag^{-1}(\mathbf{N}_h)A\circ B^\diamond\diag^{-1}(\mu_v+\delta_v)\frac{}{}\frac{}{}\right),
\end{multline*}
where $\rho(.)$ denotes the spectral radius operator.
\end{lemma}
\begin{proof}\hfill

\noindent Following the next generation method \cite{MR1057044,VddWat02}, the system composed of the infected variables in (\ref{MultiHostMultiVectorCompact}) could be decomposed as the sum of two columns vectors  $\mathcal F(\mathbf{E},\mathbf{I})$ and $\mathcal V(\mathbf{E},\mathbf{I})$ where $\mathcal F$ represents new infections in each host and arthropod species and $\mathcal V$ that of transitions between classes. By letting $F=D\mathcal F(\mathbf{E},\mathbf{I})|_{E_0}$ and $V=D\mathcal V(\mathbf{E},\mathbf{I})|_\mathrm{E_0}$, the Jacobian of $\mathcal F$ and $\mathcal V$, evaluated at the DFE, we obtain :
$$F=\left(\begin{array}{ccccc}
\block(1,2){\textbf{0}_{m(n+1),m(n+1)} }& A\circ B^\diamond\\
\textbf{0}_{p,m} &\diag(\mathbf{N}_v)\mathcal B \diag^{-1}(\mathbf{N}_h)& \textbf{0}_{mn+p,p}
\end{array}\right)$$
where $$\mathcal B=\left(\begin{array}{cccc}(A\circ B_1)^T& (A\circ B_2)^T&\dots &(A\circ B_n)^T\end{array}\right),$$
and
$$
V 
=\left(\begin{array}{ccccc}
-\textrm{diag}( \mu+\nu)& \textbf{0}_{m,m}& \hdots &\textbf{0}_{m,m}&\textbf{0}_{m,p}\\ 
\textrm{diag}(\nu)  & -\textrm{diag}(\alpha_1) &  \hdots&\textbf{0}_{m,m}&\textbf{0}_{m,p}\\
\textbf{0}_{m,m}& \textrm{diag}(\gamma_1)   &\ddots&\vdots&\vdots\\
\vdots & \ddots & \ddots&  \vdots&\vdots\\
\textbf{0}_{m,m}& \textbf{0}_{m,m}&\textrm{diag}(\gamma_{n-1})  & -\textrm{diag}(\alpha_n)&\textbf{0}_{m,p}\\
\textbf{0}_{p,p}& \textbf{0}_{p,p}&\textbf{0}_{p,p}&\textbf{0}_{p,p}  & -\textrm{diag}(\mu_v+\delta_v)\\
\end{array}\right).$$
\noindent It is worth noticing that the matrix $\mathcal B$ has $m\times n$ columns and so, the last row of $F$ has $mn+m+p=m(n+1)+p$ columns, in accordance with the first row of $F$, as $A\circ B^\diamond$ has $p$ columns. The matrix $V$ is Metzler (positive off-diagonal entries) and therefore $-V^{-1}\geq0$ \cite{0815.15016}.
Given the particular shape of the matrix $V$, the matrix-wise entries of $-V^{-1}$ are given by the induction relationship:
$$(-V^{-1})_{ij}=
\left\{\begin{array}{c}
\textbf{0}_{m,m} \quad\textrm{if}\quad 1\leq i\leq n+1,\; 1\leq i<j\leq n+1\\
\textrm{diag}^{-1}(\alpha_{i-1}) \quad\textrm{if}\quad 1\leq i=j\leq n+1,\\
\textrm{diag}(\gamma_{i-2})\textrm{diag}^{-1}(\alpha_i)(-V^{-1})_{i-1,j} \quad\textrm{if}\quad  2\leq i\leq n, \;1\leq j<i\leq n,\\
\textbf{0}_{p,p} \quad\textrm{if}\quad i=n+2,\; 1\leq j\leq n+1,\\
\textbf{0}_{p,p} \quad\textrm{if}\quad 1\leq i\leq n+1,\;  j=n+2,\\
\textrm{diag}^{-1}(\mu_v+\delta_v) \quad\textrm{if}\quad i=j=n+2,\\
\end{array}\right.
$$
where $\alpha_0=\mu+\nu$ and  $\gamma_0=\nu$. Hence, we can deduce that
{\small{
$$(-V^{-1})_{ij}=
\left\{\begin{array}{llllll}
\textbf{0}_{m,m} \quad\textrm{if}\quad 1\leq i\leq n+1,\; 1\leq i<j\leq n+1\\
\textrm{diag}^{-1}(\alpha_{i-1}) \quad\textrm{if}\quad 1\leq i=j\leq n+1,\\
\diag(\gamma_{j-1}\circ\dots\circ\gamma_{i-2})\diag^{-1}(\alpha_{j-1}\circ\dots\circ\alpha_{i-1})\;\textrm{if}\;2\leq i\leq n, \;1\leq j<i\leq n\\
\textbf{0}_{p,p} \quad\textrm{if}\quad i=n+2,\; 1\leq j\leq n+1,\\
\textbf{0}_{p,p} \quad\textrm{if}\quad  1\leq i\leq n+1,\; j=n+2,\\
\textrm{diag}^{-1}(\mu_v+\delta_v) \quad\textrm{if}\quad i=j=n+2.\\
\end{array}\right.
$$
}}
Hence, the next generation matrix is given by:

\[
-FV^{-1}=\left(\begin{array}{ccccc|c}
 &   &  & & &A\circ B^\diamond\diag^{-1}(\mu_v+\delta_v)\\ 
\multicolumn{5}{c|}{\multirow{3}{*}{\raisebox{-7mm}{\scalebox{2}{$\textbf{0}_{m(n+1),m(n+1)}$}}}} &\textbf{0}_{m,p}\\
& & & & &\vdots\\
  &   &  &  & &\vdots\\
 & & &   & &\textbf{0}_{m,p}\\
\hline
\Theta_1&\Theta_2&\Theta_3&\hdots&\Theta_{n+1}  & \textbf{0}_{p,p}\\
\end{array}\right)
\]
\noindent where 
$\Theta_1=\diag(\mathbf{N}_v)\left(\dsum_{l=1}^n(A\circ B_l)^T (-V^{-1})_{l+1,1}\right)\diag^{-1}(\mathbf{N}_h)$ and, 
$$\Theta_{k+1}=\diag(\mathbf{N}_v)\left(\dsum_{l=k}^n(A\circ B_l)^T (-V^{-1})_{l+1,k+1}\right)\diag^{-1}(\mathbf{N}_h)\;\;\textrm{for}\;\; k=1,2,\dots,n.$$
The basic reproduction number is the spectral of the next generation matrix. Hence, by denoting $\mathcal R_0(m,n,p)$ the basic reproduction number of System (\ref{MultiHostMultiVectorCompact}) with $n$ hosts, $n$ infectious stages (in each host) and $p$ vectors, we have: $\mathcal R_0(m,n,p)=\rho(-FV^{-1})$. The matrix $-FV^{-1}$ is an anti-diagonal block matrix, with $\Theta_1$ and $A\circ B^\diamond\diag^{-1}(\mu_v+\delta_v)$ the anti-diagonal entries. Therefore,
\begin{eqnarray}
 \mathcal{R}_0^2(m,n,p)&=&\rho((-FV^{-1})^2)\nonumber\\
 &=&\rho\left(\Theta_1A\circ B^\diamond\diag^{-1}(\mu_v+\delta_v)\right)\nonumber\\
 &=&\rho\left(\diag(\mathbf{N}_v)\left(\dsum_{l=1}^n(A\circ B_l)^T (-V^{-1})_{l+1,1}\right)\diag^{-1}(\mathbf{N}_h)A\circ B^\diamond\diag^{-1}(\mu_v+\delta_v)\right)\nonumber\\
  &=&\rho\left(\frac{}{}\diag(\mathbf{N}_v)\left(\dsum_{l=1}^n(A\circ B_l)^T \diag(\gamma_1\circ\gamma_2\circ\dots\circ\gamma_{l-1}\circ\nu)\diag^{-1}(( \mu+\nu)\right.\right.\nonumber\\
  & &\left.\left.\circ\alpha_1\circ\alpha_2\circ\dots\circ\alpha_{l})\right)\diag^{-1}(\mathbf{N}_h)A\circ B^\diamond\diag^{-1}(\mu_v+\delta_v)\frac{}{}\right).\nonumber
 \end{eqnarray}
This proves our claim.
\end{proof}
\noindent In the next section, we study the asymptotic properties of System (\ref{MultiHostMultiVectorCompact}). Particularly, we show that the dynamics of the system is completely determined by $ \mathcal{R}_0^2(m,n,p)$ under certain conditions. 
\section{Global stability of equilibria}\label{GASEquilibria}
\subsection{The Disease Free Equilibrium}
The disease-free equilibrium (DFE) always exists and is in $\Omega$. The following theorem gives conditions under which it is globally asymptotically stable (GAS).
\begin{theorem}\label{GASDFE}
The DFE is GAS whenever $ \mathcal R_0^2(m,n,p)<1$.
\end{theorem}
\begin{proof}
\hfill\\
The proof consists of proving that all infected variables of System (\ref{MultiHostMultiVector}) converge to zero and using the local stability of the DFE when $ \mathcal R_0^2(m,n,p)<1$ to conclude the GAS of the disease-free steady state.
Considering that $\diag(\mathbf{S}_h)\leq\diag(\mathbf{N}_h)$ and
\begin{eqnarray}
\displaystyle\dot{\mathbf{I}}_v&=&\diag(\mathbf{S}_v)\sum_{l=1}^n(A\circ B_l)^T\diag^{-1}(\mathbf{N}_h)\mathbf{I}_l- \diag(\mu_v+\delta_v)\mathbf{I}_v\nonumber\\
&=&\diag(\mathbf{N}_v-\mathbf{I}_v)\sum_{l=1}^n(A\circ B_l)^T\diag^{-1}(\mathbf{N}_h)\mathbf{I}_l- \diag(\mu_v+\delta_v)\mathbf{I}_v\nonumber\\
&\leq&\diag(\mathbf{N}_v)\sum_{l=1}^n(A\circ B_l)^T\diag^{-1}(\mathbf{N}_h)\mathbf{I}_l- \diag(\mu_v+\delta_v)\mathbf{I}_v,\nonumber
\end{eqnarray}
we obtain, from System (\ref{MultiHostMultiVectorCompact}), that
$$ 
\left\{\begin{array}{llll}%
\dot{\mathbf{E}}\leq A\circ B^\diamond \mathbf{I}_v-\diag(\mu_h+\nu_h)\mathbf{E}\\
\dot{\mathbf{I}}_1=\diag(\nu)\mathbf{E}-\diag(\alpha)\mathbf{I}_1\\
\dot{\mathbf{I}}_2=\diag(\gamma_{1})\mathbf{I}_1-\diag(\alpha_1)\mathbf{I}_2\\ 
\vdots\\
\dot{\mathbf{I}}_{n-1}=\diag(\gamma_{n-2})\mathbf{I}_{n-2}-\diag(\alpha_{n-1})\mathbf{I}_{n-1}\\
\dot{\mathbf{I}}_{n}=\diag(\gamma_{n-1})\mathbf{I}_{n-1}-\diag(\alpha_{n})\mathbf{I}_{n}\\
\displaystyle\dot{\mathbf{I}}_v\leq\diag(\mathbf{N}_v)\sum_{l=1}^n(A\circ B_l)^T\diag^{-1}(\mathbf{N}_h)\mathbf{I}_l- \diag(\mu_v+\delta_v)\mathbf{I}_v
\end{array}\right.
$$ 
Therefore, we deduce that
\begin{equation}
\label{Compa2}
\left(\begin{array}{c} 
\dot{\mathbf{E}}\\
\dot{\mathbf{I}}_1\\
\vdots\\
\dot{\mathbf{I}}_n\\
\dot{\mathbf{I}}_v\\
\end{array}\right)\leq (F+V)\left(\begin{array}{c}
\mathbf{E}\\
\mathbf{I}_1\\
\vdots\\
\mathbf{I}_n\\
\mathbf{I}_v\\
\end{array}\right)
\end{equation}
where $F$ and $V$ are the matrices generated in the next generation method. Since $F$ is a nonnegative matrix and $V$ is a Metzler matrix, we have (see \cite{0815.15016}), 
$$\rho(-FV^{-1})<1\iff \alpha(F+V)<0$$
where $\alpha(F+V)$ is the stability modulus of $F+V$. Hence, the trajectories of the auxiliary system whose RHS is that of (\ref{Compa2}) converge to zero whenever $ \mathcal R_0(m,n,p)=\rho(-FV^{-1})<1$. Since all the variables are positive, by the comparison theorem \cite{cite-SmitWalt95}, we conclude that
 \begin{equation}\label{Ito0}
 \lim_{t\to\infty}\mathbf{E}=\lim_{t\to\infty}\mathbf{I}_1=\dots=\lim_{t\to\infty}\mathbf{I}_n=\mathbf{0}_m\quad\textrm{and}\quad\lim_{t\to\infty}\mathbf{I}_v=\mathbf{0}_p.
 \end{equation}
Moreover, this implies that:
 \begin{equation}\label{StoSstar}
 \lim_{t\to\infty}\mathbf{S}=\mathbf{\Lambda}_h\oslash{\mu_h}\quad\textrm{and}\quad\lim_{t\to\infty}\mathbf{S}_v=\mathbf{\Lambda}_v\oslash{(\mu_h+\delta_v)},
 \end{equation}
 where $\oslash$ denotes the Hadamard division. Relations (\ref{Ito0}) and (\ref{StoSstar}) imply the attractivity of the DFE. Moreover, by \cite{MR1057044,VddWat02}, the DFE is locally asymptotically stable whenever $\mathcal R_0^2(m,n,p)<1$. We conclude thus that the DFE is GAS whenever $\mathcal R_0^2(m,n,p)<1$.   
\end{proof}
\begin{remark}
In \Cref{GASDFE}, it is worthwhile noting that no hypothesis on the irreducibility of the connectivity matrix is needed. This is important as this hypothesis is customarily used \cite{bichara2017multi,Guo_li_CAMQ_06,iggidr2014dynamics,shuai2011global} to obtain a Lyapunov function in the form of $V=c^TI$ where $c$ is the left eigenvector of the next generation matrix.
\end{remark}
 \subsection{Existence and Uniqueness of the EE}
 In this subsection, we explore the asymptotic behavior of System (\ref{MultiHostMultiVectorCompact}) when $\mathcal{R}^2_0(m,n,p)>1$. Particularly, we want to obtain conditions for which the disease persists in all stages, of all hosts and all vectors. That is, an equilibrium such that $$I_{l,i}^\ast\gg0\quad I_{v,j}^\ast\gg0\quad\textrm{for all}\quad l,i,j.$$
 Such interior equilibrium is also called \textit{strongly} endemic \cite{MR1993355}. The existence of such equilibrium is tied to the overall basic reproduction number and the connectivity between host and vector species. That is, the matrix $$\mathcal N=\left(\begin{array}{cccc}
\textbf{0}_{m(n+1),m(n+1)} & A\circ B^\diamond\\
\dsum_{l=1}^n(A\circ B_l)^T  & \textbf{0}_{mn+p,p}
\end{array}\right).$$
The following theorem gives the existence conditions of such an equilibrium.
 
 \begin{theorem}\label{ExistenceUniquenessEE}
 There exists a unique strongly endemic equilibrium for Model (\ref{MultiHostMultiVectorCompact}) whenever $\mathcal{R}^2_0(m,n,p)>1$ and the Host-Vector connectivity configuration  $\mathcal N$ is irreducible.
 \end{theorem}
\begin{proof} 
An endemic equilibrium satisfies the relations:
  \begin{equation} \label{EErelations}
\left\{\begin{array}{llll}
\mathbf{\Lambda}_{h}=  \diag^{-1}(\mathbf{N}_h)\diag(\mathbf{S}^*)A\circ B^\diamond \mathbf{I}_v^*+\diag(\mu_h) \mathbf{S}^*\\
\diag(\mu_h+\nu_h)\mathbf{E}^*=\diag^{-1}(\mathbf{N}_h)\diag(\mathbf{S}^*)A\circ B^\diamond \mathbf{I}_v^*\\
\diag(\alpha_1)\mathbf{I}_1^*=\diag(\nu)\mathbf{E}^*\\
\diag(\alpha_2)\mathbf{I}_2^*=\diag(\gamma_{1})\mathbf{I}_1^*\\ 
\vdots\\
\diag(\alpha_{n-1})\mathbf{I}_{n-1}^*=\diag(\gamma_{n-2})\mathbf{I}_{n-2}^*\\
\diag(\alpha_{n})\mathbf{I}_{n}^*=\diag(\gamma_{n-1})\mathbf{I}_{n-1}^*\\
 \mathbf{\Lambda}_{v}= \diag(\mathbf{S}_v^*)\sum_{l=1}^n(A\circ B_l)^T\diag^{-1}(\mathbf{N}_h)\mathbf{I}_l^*+ \diag(\mu_v+\delta_v)\mathbf{S}_v^*\\
\diag(\mu_v+\delta_v)\mathbf{I}_v^*=\diag(\mathbf{S}_v^*)\sum_{l=1}^n(A\circ B_l)^T\diag^{-1}(\mathbf{N}_h)\mathbf{I}_l^*  
\end{array}\right.
\end{equation}
In the following, we express all variables at equilibrium in terms of $\mathbf{I}_v^*$. From (\ref{EErelations}), we express all $
\mathbf{I}_l^*$ ($1\leq l\leq n$) in terms of $\mathbf{I}_1^*$, as follows:
\begin{equation}\label{Il}
\mathbf{I}_l^*=\diag(\gamma_1\circ\gamma_2\circ\dots\circ\gamma_{l-1})\diag^{-1}(\alpha_2\circ\dots\circ\alpha_{l})\mathbf{I}_1^*.
\end{equation}
From (\ref{Il}), it follows that $\mathbf{I}_l^*\gg0$, for any $1\leq l\leq n$, if and only $\mathbf{I}_1^*\gg0$. Moreover, we can also express $\mathbf{I}_1^*$ in terms of $\mathbf{I}_v^*$ and $\mathbf{S}^*$ since
\begin{eqnarray}\label{I1inTermsofIv}
\mathbf{I}_1^*&=&\diag^{-1}(\alpha_1)\diag(\nu)\mathbf{E}^*\nonumber\\
& =&\diag^{-1}(\alpha_1)\diag^{-1}(\mu_h+\nu_h)\diag(\nu)\diag^{-1}(\mathbf{N}_h)\diag(\mathbf{S}^*)A\circ B^\diamond \mathbf{I}_v^*
\end{eqnarray}
Using the equilibrium relation stemming from the equation of infected vectors (in \ref{EErelations}), the relation $\mathbf{S}_v^*=\mathbf{N}_v-\mathbf{I}_v^*$, along with (\ref{Il}) and (\ref{I1inTermsofIv}). we obtain
\begin{equation}\label{IvinTermsofSh}
\diag(\mu_v+\delta_v)\mathbf{I}_v^*=\diag(\mathbf{N}_v-\mathbf{I}_v^*)\sum_{l=1}^n(A\circ B_l)^TD_l\diag(\mathbf{S}^*)A\circ B^\diamond \mathbf{I}_v^*,
\end{equation}
where
$$D_l=\diag^{-1}(\mathbf{N}_h)\diag(\gamma_0\circ\gamma_1\circ\gamma_2\circ\dots\circ\gamma_{l-1})\diag^{-1}(\alpha_0\alpha_1\circ\alpha_2\circ\dots\circ\alpha_{l})\diag^{-1}(\mathbf{N}_h)$$
and $\alpha_0=\mu_h+\nu_h$ and $\gamma_0=\nu$. Now, we want to express $\mathbf{S}$ in terms of $\mathbf{I}_v^*$. That will make  (\ref{IvinTermsofSh}) in terms of $\mathbf{I}_v^*$ only. By using the equation of susceptible in (\ref{EErelations}), we obtain
$$
\diag(\mathbf{S}^*)=\diag(\mathbf{N}_h)\diag^{-1}(A\circ B^\diamond\mathbf{I}^*_v)+\diag^{-1}(\mu).
$$
Hence, Equation (\ref{IvinTermsofSh}) leads to:
\begin{eqnarray}\label{IvOnly}
\diag(\mu_v+\delta_v)\mathbf{I}_v^*&=&\diag(\mathbf{N}_v-\mathbf{I}_v^*)\sum_{l=1}^n(A\circ B_l)^TD_l[\diag(\mathbf{N}_h)\diag^{-1}(A\circ B^\diamond\mathbf{I}^*_v)+\diag^{-1}(\mu)]A\circ B^\diamond \mathbf{I}_v^*\nonumber\\
&=&
\diag(\mathbf{N}_v-\mathbf{I}_v^*)\sum_{l=1}^n(A\circ B_l)^TD_l\diag(\mathbf{N}_h)\diag^{-1}(A\circ B^\diamond\mathbf{I}^*_v+\mu\circ N_h)A\circ B^\diamond \mathbf{I}_v^*\nonumber\\
&:=&F(\mathbf{I}_v^*)\nonumber\\
\end{eqnarray}
Thus, a unique strongly endemic equilibrium exists if and only if the vectorial equation (\ref{IvOnly}) has a unique nonnegative solution. Moreover, notice that Equation (\ref{IvOnly}) is satisfied if and only if $\mathbf{I}_v^*$ is an equilibrium of the auxiliary dynamical system:
 \begin{equation}\label{AuxDS}
\dot x=F(x)-\diag(\mu_v+\delta_v)x.
\end{equation}
Now, we will show that System (\ref{AuxDS}) has unique equilibrium $\mathbf{I}_v^*\gg0$ if the connectivity matrix $\mathcal N$ is irreducible and $\mathcal R^2_0(m,n,p)>1$. To this end, we will use elements from monotone systems and particularly Hirsch's theorem \cite{Hirsch84}. Indeed, System (\ref{AuxDS}) is monotone if and only if the vector field $F(x)-\diag(\mu_v+\delta_v)x$ is, or equivalently whenever the $F'(x)$ is Metzler. That is, the off-diagonal elements of $F'(x)$ are positive. We have:
\begin{multline}\label{Fprime}
F'(x) =\left(\diag(\mathbf{N}_v-x)\sum_{l=1}^n(A\circ B_l)^TD_l\diag(\mathbf{N}_h)\diag^{-1}(A\circ B^\diamond x+\mu\circ N_h)A\circ B^\diamond x\right)' \\
=-\diag\left( \sum_{l=1}^n(A\circ B_l)^TD_l\diag(\mathbf{N}_h)\diag^{-1}(A\circ B^\diamond x+\mu\circ N_h)A\circ B^\diamond x\right) \\
  +\diag(\mathbf{N}_v-x)\sum_{l=1}^n(A\circ B_l)^TD_l\diag(\mathbf{N}_h)\left(\frac{}{}\diag^{-1}(A\circ B^\diamond x+\mu\circ N_h)A\circ B^\diamond x\right)',
  \end{multline}
since the first term is a diagonal matrix. It follows from (\ref{Fprime}), that $F'(x)$ is Metzler as long as $M:=(\diag^{-1}(A\circ B^\diamond x+\mu\circ N_h)A\circ B^\diamond x)'$ is. We have:
\begin{eqnarray}\label{Fprime2}
M & =&(\diag^{-1}(A\circ B^\diamond x+\mu\circ N_h)A\circ B^\diamond x)'\nonumber\\
&=&\diag^{-1}(A\circ B^\diamond x+\mu\circ N_h)A\circ B^\diamond 
 +\diag(A\circ B^\diamond x)(\diag^{-1}(A\circ B^\diamond x+\mu\circ N_h))'\nonumber\\
&=&\diag^{-1}(A\circ B^\diamond x+\mu\circ N_h)A\circ B^\diamond\nonumber\\
& &+\diag(A\circ B^\diamond x)\left(-\frac{}{}\diag^{-1}(A\circ B^\diamond x+\mu\circ N_h)A\circ B^\diamond \diag^{-1}(A\circ B^\diamond x+\mu\circ N_h)\right)\nonumber\\
&=&\diag^{-1}(A\circ B^\diamond x+\mu\circ N_h)\left(\frac{}{} A\circ B^\diamond\diag(A\circ B^\diamond x+\mu\circ N_h)\right.\nonumber\\
& & -\left.\diag(A\circ B^\diamond x)A\circ B^\diamond\frac{}{}\right) \diag^{-1}(A\circ B^\diamond x+\mu\circ N_h)\nonumber\\
&=&\diag^{-1}(A\circ B^\diamond x+\mu\circ N_h) A\circ B^\diamond\diag(\mu\circ N_h)\diag^{-1}(A\circ B^\diamond x+\mu\circ N_h),\nonumber
\end{eqnarray}
as $A\circ B^\diamond\diag(A\circ B^\diamond x)=\diag(A\circ B^\diamond )A\circ B^\diamond x$. Hence, $M$ and therefore $F'(x)$ is Metzler. Moreover, since the Host-Vector configuration $\mathcal N$ is irreducible, $F'(x)$ is irreducible. Thus, the auxiliary dynamical system (\ref{AuxDS}) is strongly monotone. Moreover, it could seen that the map of $F'(x)$ is monotonically decreasing. Since  $F(0_{\R^p})-\diag(\mu_v+\delta_v)0_{\R^p}=0_{\R^p}$ and $x\leq N_v$, by Hirsch's theorem (\cite{Hirsch84}, page 55), the trajectories of the auxiliary system (\ref{AuxDS}) either tend to the origin, or else there is a unique equilibrium $\mathbf{I}_v^*\gg0$ and all trajectories tend to $\mathbf{I}_v^*$. The origin is unstable since, for $G(x):=F(x)-\diag(\mu_v+\delta_v)x$, we have $\rho(G'(0_{\R^p}))>1$. Indeed, it could be shown easily that $\rho(\diag^{-1}(\mu_v+\delta_v)F'(0_{\R^p}))=\mathcal R^2_0(m,n,p).$ We have shown that if $\mathcal R^2_0(m,n,p)>1$ and the host-vector configuration is irreducible, a unique solution $\mathbf{I}_v^*\gg0$ of  Equation (\ref{IvOnly}) exists. Hence, using (\ref{EErelations}), (\ref{Il}), we deduce that $\mathbf{E}^*\gg0$ and  $\mathbf{I}_l^*\gg0$, for $1\leq l\leq n$. We conclude therefore that
 System (\ref{MultiHostMultiVectorCompact}) has a unique strongly endemic equilibrium whenever $\mathcal{R}^2_0(m,n,p)>1$ and $\mathcal N$ is irreducible.
\end{proof}
The following subsection consists of investigating the asymptotic properties of this unique strongly endemic equilibrium.
  \subsection{Global Stability of the EE}
  \begin{theorem}\label{GASEE}\hfill 
  
The strongly endemic equilibrium is GAS whenever it exists.
 \end{theorem}
  \begin{proof}
  We consider the following Lyapunov candidate $\mathcal V=\sum_{i=1}^mv_i\mathcal V_{i}$
  where
  \begin{eqnarray*}
  \mathcal V_{i}&=&\int_{S_i^\ast}^{S_i}\left(1-\frac{S_i^\ast}{x}\right)dx+\int_{E_{i}^\ast}^{E_i}\left(1-\frac{E_i^\ast}{x}\right)dx+\sum_{k=1}^nc_{k,i}\int_{I_{k,i}^\ast}^{I_{k,i}}\left(1-\frac{I_{k,i}^\ast}{x}\right)dx+\sum_{j=1}^pw_{ij}\mathcal V_{v,j},
  \end{eqnarray*}
$$\mathcal V_{v,j}=\int_{S_{v,j}^\ast}^{S_{v,j}}\left(1-\frac{S_{v,j}^\ast}{x}\right)dx +\int_{I_{v,j}^\ast}^{I_{v,j}}\left(1-\frac{I_{v,j}^\ast}{x}\right)dx,\;\; w_{ij}=\frac{a_{ij}\beta_{ij}^\diamond S_i^\ast}{N_i(\mu_{v,j}+\delta_{v,j})},$$
$$c_{1,i}=\frac{\nu_i+\mu_i}{\nu_i},\;\;\textrm{and}\;\; c_{k,i}=\frac{1}{\alpha_{k,i}I_{k,i}^\ast}\sum_{j=1}^pw_{ij} \sum_{i=1}^ma_{ij}\frac{S_{v,j}^\ast}{N_i}\sum_{l=k}^n\beta_{l,j}^iI_{l,i}^\ast\quad \forall k=1,\dots,n$$
The coefficients $v_i$ are positive to be determined later. The function $\mathcal V$ is definite positive and the goal is show that its derivative along the trajectories of the multi-host, multi-vector system (\ref{MultiHostMultiVectorCompact}) is definite-negative. To ease the notations, let  $$\bar\beta_{l,j}^i=a_{ij}S_{v,j}^*\frac{\beta_{l,j}^iI_{l,i}^*}{N_i}.$$
Throughout this proof, we will be using the component-wise endemic relations, which could be written as:
\begin{equation} \label{EErelationsCom}
\left\{\begin{array}{llll} 
\displaystyle\Lambda_i=\sum_{j=1}^pa_{i,j}\beta_{i,j}^\diamond S_{i}^\ast\dfrac{I_{v,j}^\ast}{N_{i}}+\mu_i S_i^\ast\\
(\mu_i+\nu_i)E_{h}^\ast=\dsum_{j=1}^pa_{i,j}\beta_{i,j}^\diamond S_{i}^\ast\dfrac{I_{v,j}^\ast}{N_{i}}\\
\alpha_{1,i} I_{1,i}^\ast=\nu_iE_i^\ast\\
\alpha_{2,i}I_{2,i}^\ast=\gamma_{1,i}I_{1,i}^\ast\\ 
\vdots\\
\alpha_{n-1,i} I_{n-1,i}^\ast=\gamma_{n-2,i}I_{n-2,i}^\ast\\
\alpha_{n,i} I_{n,i}^\ast=\gamma_{i,n-1}I_{i,n-1}^\ast\\
\displaystyle \Lambda_{v,j}=\sum_{i=1}^m\sum_{l=1}^n\bar\beta_{l,j}^i+(\mu_{v,j}+\delta_{v,j})S_{v,j}^\ast\\
\displaystyle (\mu_{v,j}+\delta_{v,j})I_{v,j}^\ast=\sum_{i=1}^m\sum_{l=1}^n \bar\beta_{l,j}^i
\end{array}\right.
\end{equation}
The derivative of $\mathcal V$ along the trajectories of (\ref{MultiHostMultiVectorCompact}) is given by:
  \begin{multline}\label{LyapDotHostGeneral111}
\dot{\mathcal V}=  \sum_{i=1}^mv_i\left[\mathcal A_{h,i}+\sum_{j=1}^pa_{i,j}\beta_{i,j}^\diamond S_{i}^\ast\dfrac{I_{v,j}^\ast}{N_{i}}\left( 2-\frac{S_i^\ast}{S_i} -\frac{E_i^\ast}{E_{i}}\frac{S_{i}}{S_{i}^\ast}\dfrac{I_{v,j}}{I_{v,j}^\ast}
-\frac{E_{i}}{E_{i}^\ast}\frac{I_{1,i}^\ast}{I_{1,i}}
\right) 
\right. \\
 \left. +\sum_{j=1}^pa_{i,j}\beta_{i,j}^\diamond S_{i}^\ast\dfrac{I_{v,j}}{N_{i}} -\left(\mu_i+\nu_i-c_{1,i}\nu_i\right)E_i   
 -\sum_{k=2}^nc_{k,i}\alpha_{k,i}I_{k,i}^\ast\frac{I_{k-1,i}}{I_{k-1,i}^\ast}\frac{I_{k,i}^\ast}{I_{k,i}}\right. \\
  \left.  
  +\sum_{k=1}^nc_{k,i}\alpha_{k,i}I_{k,i}^\ast 
-\sum_{k=1}^{n-1}(c_{k,i}\alpha_{k,i}-c_{k+1,i}\gamma_{k,i})I_{k,i}-c_{n,i}\alpha_{n,i}I_{n,i}\right. \\
 %
 \left. 
+ \sum_{j=1}^pw_{ij}\mathcal A_{v,j} 
+ \dsum_{j=1}^pw_{ij}\sum_{i=1}^m \sum_{l=1}^n \bar\beta_{l,j}^i  \left( 2 -\frac{S_{v,j}^\ast}{S_{v,j}}
 -\frac{I_{v,j}^\ast}{I_{v,j}}   \frac{S_{v,j}}{S_{v,j}^\ast}\frac{I_{l,i}}{I_{l,i}^\ast}
 \right)  \right. \\
  \left. +\dsum_{j=1}^pw_{ij}\sum_{i=1}^m \sum_{l=1}^n\bar\beta_{l,j}^i\dfrac{ I_{l,i}}{I_{l,i}^*}%
 -\dsum_{j=1}^pw_{ij}(\mu_{v,j}+\delta_{v,j})I_{v,j}\right],
    \end{multline}
where $\displaystyle \mathcal A_{h,i}=\mu_iS_i^\ast\left(2 -\frac{S_i^\ast}{S_i} -\frac{S_i}{S_i^\ast} \right)$ and $\displaystyle \mathcal A_{v,j}=(\mu_{v,j}+\delta_{v,j})S_{v,j}^\ast\left( 2-\frac{S_{v,j}^\ast}{S_{v,j}}-\frac{S_{v,j}}{S_{v,j}^\ast} \right).$

Given the expression of $w_{ij}$, the terms in $I_{v,j}$ sum to zero. Similarly, using the expression of $c_{1,i}$, the terms in $E_i$ also sum to zero.
  Hence, what remains in (\ref{LyapDotHostGeneral111}) is:
 \begin{multline}\label{LyapDotHostGeneral11}
\dot{\mathcal V}  =\sum_{i=1}^mv_i\left[\sum_{j=1}^pa_{i,j}\beta_{i,j}^\diamond S_{i}^\ast\dfrac{I_{v,j}^\ast}{N_{i}}\left( 2-\frac{S_i^\ast}{S_i} -\frac{E_i^\ast}{E_{i}}\frac{S_{i}}{S_{i}^\ast}\dfrac{I_{v,j}}{I_{v,j}^\ast}
-\frac{E_{i}}{E_{i}^\ast}\frac{I_{1,i}^\ast}{I_{1,i}}
\right) 
\right. \\
   \left.
  -\sum_{k=2}^nc_{k,i}\alpha_{k,i}I_{k,i}^\ast\frac{I_{k-1,i}}{I_{k-1,i}^\ast}\frac{I_{k,i}^\ast}{I_{k,i}} 
  +\sum_{k=1}^nc_{k,i}\alpha_{k,i}I_{k,i}^\ast 
-\sum_{k=1}^{n-1}(c_{k,i}\alpha_{k,i}-c_{k+1,i}\gamma_{k,i})I_{k,i}\right. \\
 \left. -c_{n,i}\alpha_{n,i}I_{n,i} 
+ \dsum_{j=1}^pw_{ij}\sum_{i=1}^m \sum_{l=1}^n \bar\beta_{l,j}^i \left( 2 -\frac{S_{v,j}^\ast}{S_{v,j}}
 -\frac{I_{v,j}^\ast}{I_{v,j}}   \frac{S_{v,j}}{S_{v,j}^\ast}\frac{I_{l,i}}{I_{l,i}^\ast}
 \right)  \right. \\
 \left. +\dsum_{j=1}^pw_{ij}\sum_{i=1}^m \sum_{l=1}^n\bar\beta_{l,j}^i\dfrac{ I_{l,i}}{I_{l,i}^*}%
 +\mathcal A_{h,i}+ \sum_{j=1}^pw_{ij}\mathcal A_{v,j}
 \right]
   \end{multline} 
Our first major claim is:
$$
\sum_{i=1}^mv_i\left[
\sum_{k=1}^{n-1}(c_{k,i}\alpha_{k,i}-c_{k+1,i}\gamma_{k,i})I_{k,i}-c_{n,i}\alpha_{n,i}I_{n,i} 
 -\dsum_{j=1}^pw_{ij}\sum_{i=1}^m\sum_{l=1}^n\bar\beta_{l,j}^i \dfrac{I_{l,i}}{I_{l,i}^*}%
 \right]=0\;\; (\bigstar)
$$
Indeed, the equality $(\bigstar)$ claims that all linear terms in $I_{l,i}$ ($1\leq l\leq n$) in (\ref{LyapDotHostGeneral11}) sum to zero.  By using the expressions of $c_{k,i}$, we obtain:
\begin{eqnarray*}
c_{k,i}\alpha_{k,i}-c_{k+1,i}\gamma_{k,i}&=& \frac{1}{I_{k,i}^\ast}\sum_{j=1}^pw_{ij}\sum_{i=1}^m \sum_{l=k}^n\bar\beta_{l,j}^i 
-\frac{\gamma_{k,i}}{\alpha_{k+1,i}I_{k+1,i}^\ast}\sum_{j=1}^pw_{ij}\sum_{i=1}^m \sum_{l=k+1}^n\bar\beta_{l,j}^i\\
&=& \frac{1}{I_{k,i}^\ast}\sum_{j=1}^pw_{ij}\sum_{i=1}^m \left(\sum_{l=k}^n\bar\beta_{l,j}^i-\sum_{l=k+1}^n\bar\beta_{l,j}^i\right)\\
%
&=& \frac{1}{I_{k,i}^\ast}\sum_{j=1}^pw_{ij}\sum_{r=1}^m \bar\beta_{k,j}^r. 
\end{eqnarray*}
Substituting this expression of $c_{k,i}\alpha_{k,i}-c_{k+1,i}\gamma_{k,i}$ in $(\bigstar)$, we obtain:
$$
\sum_{i=1}^mv_i\left[
\sum_{k=1}^{n}\left(  \frac{1}{I_{k,i}^\ast}\sum_{j=1}^pw_{ij}\sum_{r=1}^m \bar\beta_{k,j}^r \right)I_{k,i} 
 -\dsum_{j=1}^pw_{ij}\sum_{r=1}^m \sum_{l=1}^n\bar\beta_{l,j}^r\dfrac{ I_{l,r}}{I_{l,r}^*}%
 \right]=0. \;\; (\bigstar\bigstar)
$$
We work the second sum in $(\bigstar\bigstar)$ out in order to obtain a factor of $I_{k,i}$ and therefore compare it with the first one.
 We have:
\begin{eqnarray*}
\sum_{i=1}^mv_i\left[
\dsum_{j=1}^pw_{ij}\sum_{r=1}^m\sum_{l=1}^n\bar\beta_{l,j}^r\dfrac{I_{l,r}}{I_{l,r}^*}%
 \right]
 &=&\sum_{i=1}^mv_i\left[\dsum_{j=1}^pw_{ij}\sum_{l=1}^n\left(\sum_{r=1}^m \beta_{l,j}^r \dfrac{I_{l,r}}{I_{l,r}^*}\right)%
 \right]\\
 &=&\sum_{i=1}^m\left[\dsum_{j=1}^p\sum_{k=1}^n\sum_{r=1}^m v_iw_{ij}\bar\beta_{k,j}^r  \dfrac{I_{k,r}}{I_{k,r}^*}%
 \right],
 \end{eqnarray*}
 obtained by replacing the index $l$ by $k$ for convenience. Now, by successively switching the subindices $i$ and $r$; and using properties of nested sums, we obtain:
 \begin{eqnarray*}
\sum_{i=1}^mv_i\left[
\dsum_{j=1}^pw_{ij}\sum_{r=1}^m\sum_{l=1}^n\bar\beta_{l,j}^r\dfrac{I_{l,r}}{I_{l,r}^*}%
 \right] 
  &=&\sum_{r=1}^m\left[\dsum_{j=1}^p\sum_{k=1}^n\sum_{i=1}^m v_rw_{rj} \bar\beta_{k,j}^i \frac{I_{k,i}}{I_{k,i}^*}%
 \right]\\
  &=&\sum_{i=1}^m\sum_{k=1}^n\left[\dsum_{r=1}^m\sum_{j=1}^p v_rw_{rj} \bar\beta_{k,j}^i %
 \right]\frac{I_{k,i}}{I_{k,i}^*}.
\end{eqnarray*}
Thus, showing $(\bigstar\bigstar)$ is equivalent to show that:
$$\dsum_{r=1}^m\sum_{j=1}^pv_rw_{rj}\bar\beta_{k,j}^i = v_i\sum_{j=1}^pw_{ij}\sum_{r=1}^m \bar\beta_{k,j}^r
$$
This is also equivalent to:
$$\dsum_{r=1}^m\sum_{j=1}^p v_rw_{rj}  \dsum_{k=1}^n \bar\beta_{k,j}^i  = v_i\sum_{j=1}^pw_{ij}\sum_{r=1}^m  \sum_{k=1}^n \bar\beta_{k,j}^r \quad(\bigstar\bigstar\bigstar)
$$
Showing $(\bigstar\bigstar\bigstar)$ concludes the proof of $(\bigstar)$.
We notice that the relation $(\bigstar\bigstar\bigstar)$ is satisfied if $v_i$ are the components of the solution of the linear system $\bar Bv=0$, with
\begin{equation}\label{Bbar}\bar B=\left(\begin{array}{cccc}
\clubsuit_{11} & \clubsuit_{12} & \dots &\clubsuit_{1m}\\
\clubsuit_{21} &
\clubsuit_{22} &
  \dots & \clubsuit_{2m}\\
\vdots &\vdots&\ddots&\vdots\\
\clubsuit_{m1} &
\clubsuit_{m2}  & \dots &
 \clubsuit_{mm} 
\end{array}\right),\end{equation}
where
%
%
\begin{eqnarray*}
\clubsuit_{kk}&=&
\dsum_{j=1}^p w_{kj}\left(-\sum_{i=1,i\neq k}^m\sum_{l=1}^n\bar\beta_{l,j}^i \right)
\end{eqnarray*}
and 
\begin{eqnarray*}\clubsuit_{ik}
&=&
\dsum_{j=1}^p w_{kj}\sum_{l=1}^n\bar\beta_{l,j}^i, \quad\forall k\neq i.\end{eqnarray*}
Moreover, we have:
\begin{eqnarray*}
\sum_{i=1}^m\clubsuit_{ik}&=&\clubsuit_{kk}+\sum_{i=1,i\neq k}^m\clubsuit_{ik}\\
&=&\dsum_{j=1}^p w_{kj}\left(-\sum_{i=1,i\neq k}^m\sum_{l=1}^n\bar\beta_{l,j}^i \right)
+\sum_{i=1,i\neq k}^m\dsum_{j=1}^p w_{kj}\sum_{l=1}^n\bar\beta_{l,j}^i\\
&=&0.
\end{eqnarray*}
Since the Host-Vector connectivity configuration $\mathcal N$ is irreducible, the matrix $\bar B$ is also irreducible. Thus, it could be shown easily that $\dim(\ker(\bar B))=1$. Moreover, $v_i=-C_{ii}$ where $C_{ii}$ is the cofactor of the $i^{th}$ diagonal of $\bar B$. Also, $C_{ii}<0$ for all $i=1,2,\dots,m$. Thus, there exists $v=(v_1,v_2,\dots,v_m)^T\gg0$ such that $\bar Bv=0$. We choose the coefficients $v_i$ of $\mathcal V_i$ in Lyapunov function are as such. \\
To show $(\bigstar\bigstar\bigstar)$, we start by its left hand side (LHS):
\begin{eqnarray*}
LHS^{(\bigstar\bigstar\bigstar)}&=&\dsum_{r=1}^m\sum_{j=1}^p\left(v_rw_{rj} \sum_{k=1}^n\bar\beta_{k,j}^i \right)\\
&=&\dsum_{r=1,r\neq i}^mv_r\sum_{j=1}^p\left(w_{rj} \sum_{k=1}^n\bar\beta_{k,j}^i\right)+
v_i\sum_{j=1}^p\left(w_{ij}  \sum_{k=1}^n\bar\beta_{k,j}^i  \right).
\end{eqnarray*}
\noindent However, $\bar Bv=0$ implies that, for any $i$, 
 \mbox{$\clubsuit_{ii}v_i+\dsum_{r=1,r\neq i}^m\clubsuit_{ir}v_r=0$}. Therefore,
$$
\dsum_{j=1}^p w_{ij}\left(-\sum_{r=1,r\neq i}^m \sum_{l=1}^n\bar\beta_{l,j}^r \right)v_i+\dsum_{r=1,r\neq i}^m\left(\dsum_{j=1}^p w_{rj}  \sum_{l=1}^n\bar\beta_{l,j}^i \right)v_r=0
$$
Thus, using the properties of nested again, $LHS^{(\bigstar\bigstar\bigstar)}$ becomes
\begin{eqnarray*}
LHS^{(\bigstar\bigstar\bigstar)}&=&
\dsum_{r=1,r\neq i}^mv_r\sum_{j=1}^p\left(w_{rj} \sum_{k=1}^n\bar\beta_{k,j}^i  \right)+
v_i\sum_{j=1}^p\left(w_{ij}  \sum_{k=1}^n\bar\beta_{k,j}^i  \right)\\
&=&\dsum_{j=1}^p w_{ij} \left(\sum_{r=1,r\neq i}^m \sum_{l=1}^n\bar\beta_{l,j}^r \right)v_i+
v_i\sum_{j=1}^p\left(w_{ij} \sum_{k=1}^n\bar\beta_{k,j}^i  \right)\\
&=&v_i\dsum_{j=1}^pw_{ij}\left(\sum_{r=1,r\neq i}^m \sum_{l=1}^n\bar\beta_{l,j}^r +
  \sum_{k=1}^n\bar\beta_{k,j}^i 
\right)\\
&=&v_i\dsum_{j=1}^pw_{ij}\left(\sum_{r=1}^m\sum_{l=1}^n\bar\beta_{l,j}^r  
\right),
\end{eqnarray*}
which is exactly the right hand side of $(\bigstar\bigstar\bigstar)$. 

In summary, Relation $(\bigstar)$, obtained through $(\bigstar\bigstar)$ and $(\bigstar\bigstar\bigstar)$  cancels linear terms in $I_{l,i}$, for $1\leq l\leq n$, in the expression of $\dot{\mathcal V}$ given in relation (\ref{LyapDotHostGeneral11}). Thus, the latter yields to:
  \begin{multline} \label{LyapDotHostGeneral222}
\dot{\mathcal V}  =\sum_{i=1}^mv_i\left[\sum_{j=1}^pa_{i,j}\beta_{i,j}^\diamond S_{i}^\ast\dfrac{I_{v,j}^\ast}{N_{i}}\left( 2-\frac{S_i^\ast}{S_i} -\frac{E_i^\ast}{E_{i}}\frac{S_{i}}{S_{i}^\ast}\dfrac{I_{v,j}}{I_{v,j}^\ast}
-\frac{E_{i}}{E_{i}^\ast}\frac{I_{1,i}^\ast}{I_{1,i}}
\right) 
\right. \\
   \left.
  -\sum_{k=2}^nc_{k,i}\alpha_{k,i}I_{k,i}^\ast\frac{I_{k-1,i}}{I_{k-1,i}^\ast}\frac{I_{k,i}^\ast}{I_{k,i}} 
  +\sum_{k=1}^nc_{k,i}\alpha_{k,i}I_{k,i}^\ast 
\right. \\
 \left. 
+ \dsum_{j=1}^pw_{ij}\sum_{i=1}^m\sum_{l=1}^n\bar\beta_{l,j}^i \left( 2 -\frac{S_{v,j}^\ast}{S_{v,j}}
 -\frac{I_{v,j}^\ast}{I_{v,j}}   \frac{S_{v,j}}{S_{v,j}^\ast}\frac{I_{l,i}}{I_{l,i}^\ast} 
 \right)  %
 +\mathcal A_{h,i}+ \sum_{j=1}^pw_{ij} \mathcal A_{v,j} \right].
    \end{multline} 
    By using endemic relations (\ref{EErelationsCom}) and the expression of the coefficients $c_{k,i}$, we have $$a_{i,j}\beta_{i,j}^\diamond S_{i}^\ast\dfrac{I_{v,j}^\ast}{N_{i}}=w_{ij}\sum_{r=1}^m\sum_{l=1}^n\bar\beta_{l,j}^r,$$
    and
    $$c_{k,i}\alpha_{k,i}I_{k,i}^\ast=\sum_{j=1}^pw_{ij}\sum_{i=1}^m\sum_{l=k}^n\bar\beta_{l,j}^i \quad \forall k=1,\dots,n.$$
     Substituting these expression and switching $i$ by $r$ in the last sum of (\ref{LyapDotHostGeneral222}), the expression $\dot{\mathcal V}$ becomes:
      \begin{multline} \label{LyapDotHostGeneralCorrect}
\dot{\mathcal V}   
  =\sum_{i=1}^mv_i\left[\sum_{j=1}^pw_{ij}\sum_{r=1}^m\sum_{l=1}^n\bar\beta_{l,j}^r\left( 2-\frac{S_i^\ast}{S_i} -\frac{E_i^\ast}{E_{i}}\frac{S_{i}}{S_{i}^\ast}\dfrac{I_{v,j}}{I_{v,j}^\ast}
-\frac{E_{i}}{E_{i}^\ast}\frac{I_{1,i}^\ast}{I_{1,i}}
\right) 
\right. \\
  \left.
 -\sum_{k=2}^n\sum_{j=1}^pw_{ij}\sum_{i=1}^m\sum_{l=k}^n\bar\beta_{l,j}^i\frac{I_{k-1,i}}{I_{k-1,i}^\ast}\frac{I_{k,i}^\ast}{I_{k,i}} 
  +\sum_{k=1}^n\sum_{j=1}^pw_{ij}\sum_{i=1}^m\sum_{l=k}^n\bar\beta_{l,j}^i
\right. \\
 \left. 
+ \dsum_{j=1}^pw_{ij}\sum_{r=1}^m\sum_{l=1}^n\bar\beta_{l,j}^r \left( 2 -\frac{S_{v,j}^\ast}{S_{v,j}}
 -\frac{I_{v,j}^\ast}{I_{v,j}}   \frac{S_{v,j}}{S_{v,j}^\ast}\frac{I_{l,r}}{I_{l,r}^\ast}
 \right)  %
  +\mathcal A_{h,i}+ \sum_{j=1}^pw_{ij} \mathcal A_{v,j}
 \right] \\
=\sum_{i=1}^mv_i\left[\sum_{j=1}^pw_{ij}\sum_{r=1}^m \bar\beta_{1,j}^r \left( 5-\frac{S_i^\ast}{S_i} -\frac{E_i^\ast}{E_{i}}\frac{S_{i}}{S_{i}^\ast}\dfrac{I_{v,j}}{I_{v,j}^\ast}
-\frac{E_{i}}{E_{i}^\ast}\frac{I_{1,i}^\ast}{I_{1,i}}
-\frac{S_{v,j}^\ast}{S_{v,j}}
 -\frac{I_{v,j}^\ast}{I_{v,j}}   \frac{S_{v,j}}{S_{v,j}^\ast}\frac{I_{1,r}}{I_{1,r}^\ast}
\right) 
\right. \\
 +\left. \sum_{j=1}^pw_{ij}\sum_{r=1}^m\sum_{l=2}^n\bar\beta_{l,j}^r \left( 5-\frac{S_i^\ast}{S_i} -\frac{E_i^\ast}{E_{i}}\frac{S_{i}}{S_{i}^\ast}\dfrac{I_{v,j}}{I_{v,j}^\ast}
-\frac{E_{i}}{E_{i}^\ast}\frac{I_{1,i}^\ast}{I_{1,i}}
-\frac{S_{v,j}^\ast}{S_{v,j}}
 -\frac{I_{v,j}^\ast}{I_{v,j}}   \frac{S_{v,j}}{S_{v,j}^\ast}\frac{I_{l,r}}{I_{l,r}^\ast}
\right) 
\right. \\
  \left.
 -\sum_{k=2}^n\sum_{j=1}^pw_{ij}\sum_{i=1}^m\sum_{l=k}^n\bar\beta_{l,j}^i\frac{I_{k-1,i}}{I_{k-1,i}^\ast}\frac{I_{k,i}^\ast}{I_{k,i}} 
  +\sum_{k=2}^n\sum_{j=1}^pw_{ij}\sum_{i=1}^m\sum_{l=k}^n\bar\beta_{l,j}^i
  + \mathcal A_{h,i}+ \sum_{j=1}^pw_{ij} \mathcal A_{v,j} 
 \right]
   \end{multline} 
   Also, notice that:
  \begin{eqnarray*}
    \sum_{k=2}^n\sum_{j=1}^pw_{ij}\sum_{i=1}^m\sum_{l=k}^n\bar\beta_{l,j}^i&=&
    \sum_{j=1}^pw_{ij}\sum_{i=1}^m\sum_{k=2}^n\sum_{l=k}^n\bar\beta_{l,j}^i\\
        &=&
    \sum_{j=1}^pw_{ij}\sum_{i=1}^m\sum_{l=2}^n(l-1)\bar\beta_{l,j}^i. 
  \end{eqnarray*}
  Therefore, Equation (\ref{LyapDotHostGeneralCorrect}) implies that:
 \begin{multline} \label{LyapDotHostGeneralCorrect2}
\dot{\mathcal V}    
    =\sum_{i=1}^mv_i\left[\sum_{j=1}^pw_{ij}\sum_{r=1}^m \bar\beta_{1,j}^r\left( 5-\frac{S_i^\ast}{S_i} -\frac{E_i^\ast S_{i}}{E_{i} S_{i}^\ast}\dfrac{I_{v,j}}{I_{v,j}^\ast}
-\frac{E_{i}}{E_{i}^\ast}\frac{I_{1,i}^\ast}{I_{1,i}}
-\frac{S_{v,j}^\ast}{S_{v,j}}
 - \frac{I_{v,j}^\ast S_{v,j}I_{1,r}}{I_{v,j}S_{v,j}^\ast I_{1,r}^\ast}
\right) 
\right. \\
 +\left. \sum_{j=1}^pw_{ij}\sum_{r=1}^m \sum_{l=2}^n\bar\beta_{l,j}^r\left( 4+l-\frac{S_i^\ast}{S_i} -\frac{E_i^\ast}{E_{i}}\frac{S_{i}}{S_{i}^\ast}\dfrac{I_{v,j}}{I_{v,j}^\ast}
-\frac{E_{i}}{E_{i}^\ast}\frac{I_{1,i}^\ast}{I_{1,i}}
-\frac{S_{v,j}^\ast}{S_{v,j}}
 -\frac{I_{v,j}^\ast}{I_{v,j}}   \frac{S_{v,j}}{S_{v,j}^\ast}\frac{I_{l,r}}{I_{l,r}^\ast}
\right) 
\right. \\
  \left.
 -\sum_{j=1}^pw_{ij}\sum_{i=1}^m\sum_{k=2}^n\sum_{l=k}^n\bar\beta_{l,j}^i\frac{I_{k-1,i}}{I_{k-1,i}^\ast}\frac{I_{k,i}^\ast}{I_{k,i}} 
  +\mathcal A_{h,i}+ \sum_{j=1}^pw_{ij} \mathcal A_{v,j}
 \right].
    \end{multline} 
     Moreover,
  \begin{eqnarray*}
   \sum_{k=2}^n\sum_{l=k}^n\bar\beta_{l,j}^i\frac{I_{k-1,i}}{I_{k-1,i}^\ast}\frac{I_{k,i}^\ast}{I_{k,i}}
   &=& \sum_{l=2}^n\sum_{k=l}^n\bar\beta_{k,j}^i\frac{I_{l-1,i}}{I_{l-1,i}^\ast}\frac{I_{l,i}^\ast}{I_{l,i}}\\
   &=&\sum_{l=2}^n\bar\beta_{l,j}^i \sum_{k=2}^l\frac{I_{k-1,i}}{I_{k-1,i}^\ast}\frac{I_{k,i}^\ast}{I_{k,i}},
   \end{eqnarray*}
   since $\dsum_{l=2}^n\dsum_{k=l}^nu_lv_k=\dsum_{l=2}^nv_l\dsum_{k=2}^lv_k.$ Thus, Equation (\ref{LyapDotHostGeneralCorrect2}) implies
  \begin{multline} \label{LyapDotHostGeneralCorrect4}
\dot{\mathcal V}    
    =\sum_{i=1}^mv_i\left[\sum_{j=1}^pw_{ij}\sum_{r=1}^m \bar\beta_{1,j}^r \left( 5-\frac{S_i^\ast}{S_i} -\frac{E_i^\ast S_{i}}{E_{i}S_{i}^\ast}\dfrac{I_{v,j}}{I_{v,j}^\ast}
-\frac{E_{i}}{E_{i}^\ast}\frac{I_{1,i}^\ast}{I_{1,i}}
-\frac{S_{v,j}^\ast}{S_{v,j}}
 -\frac{I_{v,j}^\ast S_{v,j}I_{1,r}}{I_{v,j} S_{v,j}^\ast I_{1,r}^\ast} 
\right) 
\right. \\
 +\left. \sum_{j=1}^pw_{ij}\sum_{r=1}^m \sum_{l=2}^n\bar\beta_{l,j}^r \left( 4+l-\frac{S_i^\ast}{S_i} -\frac{E_i^\ast S_{i}}{E_{i} S_{i}^\ast}\dfrac{I_{v,j}}{I_{v,j}^\ast}
-\frac{E_{i}}{E_{i}^\ast}\frac{I_{1,i}^\ast}{I_{1,i}}
-\frac{S_{v,j}^\ast}{S_{v,j}}
 -\frac{I_{v,j}^\ast}{I_{v,j}}   \frac{S_{v,j}}{S_{v,j}^\ast}\frac{I_{l,r}}{I_{l,r}^\ast}
\right) 
\right. \\
   \left.
 -\sum_{j=1}^pw_{ij}\sum_{r=1}^m \sum_{l=2}^n\bar\beta_{l,j}^r \sum_{k=2}^l\frac{I_{k-1,r}}{I_{k-1,r}^\ast}\frac{I_{k,r}^\ast}{I_{k,r}} 
  +\mathcal A_{h,i}+ \sum_{j=1}^pw_{ij} \mathcal A_{v,j} 
 \right].
     \end{multline} 
Finally, by combining the second and third sums in (\ref{LyapDotHostGeneralCorrect4}), we obtain:   
      \begin{align}\label{LyapDotHostGeneral5}
\dot{\mathcal V}    
  &=\sum_{i,r=1}^m\sum_{j=1}^pv_iw_{ij}\bar\beta_{1,j}^r\left( 5-\frac{S_i^\ast}{S_i} -\frac{E_i^\ast}{E_{i}}\frac{S_{i}}{S_{i}^\ast}\dfrac{I_{v,j}}{I_{v,j}^\ast}
-\frac{E_{i}}{E_{i}^\ast}\frac{I_{1,i}^\ast}{I_{1,i}}
-\frac{S_{v,j}^\ast}{S_{v,j}}
 -\frac{I_{v,j}^\ast}{I_{v,j}}   \frac{S_{v,j}}{S_{v,j}^\ast}\frac{I_{1,r}}{I_{1,r}^\ast}
\right)\\
&+ \sum_{i=1,r}^m\sum_{j=1}^pv_iw_{ij} \dsum_{l=2}^n\bar\beta_{l,j}^r\left( 4+l-\frac{S_i^\ast}{S_i} -\frac{E_i^\ast}{E_{i}}\frac{S_{i}}{S_{i}^\ast}\dfrac{I_{v,j}}{I_{v,j}^\ast}
-\frac{E_{i}}{E_{i}^\ast}\frac{I_{1,i}^\ast}{I_{1,i}}
-\frac{S_{v,j}^\ast}{S_{v,j}}
 -\frac{I_{v,j}^\ast}{I_{v,j}}   \frac{S_{v,j}}{S_{v,j}^\ast}\frac{I_{l,r}}{I_{l,r}^\ast}\right.\nonumber\\
&  \left.-\sum_{k=2}^l\frac{I_{k-1,r}}{I_{k-1,r}^\ast}\frac{I_{k,r}^\ast}{I_{k,r}}
\right) 
  +\sum_{i=1}^mv_i\mathcal A_{h,i}+ \sum_{i=1}^m\sum_{j=1}^pv_iw_{ij} \mathcal A_{v,j} 
 \nonumber\\
 &=  \mathcal S_1+\mathcal S_2+\sum_{i=1}^mv_i\mathcal A_{h,i}   +\sum_{i=1}^m\sum_{j=1}^pv_iw_{ij} \mathcal A_{v,j}\nonumber. 
    \end{align}
The terms $\mathcal A_{v,j}$ and $\mathcal A_{h,i}$ are clearly definite-negative. Note that, taken separately,  the two sums $\mathcal S_1$ and $\mathcal S_2$ in (\ref{LyapDotHostGeneral5}) are not definite negative. Hence, to finish the proof, we will prove that the sum $ \mathcal S_1+\mathcal S_2$ is definite-negative.
   To do so, we look at the coefficients of the sums from the graph-theoretical standpoint, following the same approach as  \cite{Guo_li_CAMQ_06,Guo_li_PAMS08,guo2012global,iggidr2014dynamics}. Indeed,  let $G(\mathcal N)$ be the directed graph that represents the connectivity $\mathcal N$ between the $m$ hosts (including the $l$ stages) and $p$ vectors. Since the Host-Vector connectivity configuration $\mathcal N$ is irreducible, it follows that graph $G(\mathcal N)$ is strongly connected. Recall that, for $i=1,2,\dots,m$, $v_i$ are components of the solution of the system $\bar Bv=0$, where $\bar B$ is given by (\ref{Bbar}). The matrix $\bar B$ is the so-called Laplacian matrix and its associated the graph $G(\bar B)$ is strongly connected since $\bar B$ is irreducible, since $\mathcal N$ is. Moreover, the solution of $\bar Bv=0$ is given by the Kirchhoff's matrix tree theorem \cite{bollobas2013modern,moon1970counting} as follows: $$v_i=\sum_{T\in\mathbb T_i}w(T)$$
 where $\mathbb T_i$ is the set of all spanning trees $T$ of $G(\bar B)$ rooted at Host $i$. Particularly, $$w(T)=\prod_{(m,l,r,j)\in E(T)}\bar\beta_{l,j}^mw_{rj},$$
 where $E(T)$ is the set of all arcs in $T$. For our setup, an arc $(i,l,i',j)$ describes an infection arc that starts from Host $i$, at stage $l$ directed to Host $i'$ through Vector $j$. From a modeling standpoint, $v_i$ connects Host $i$ to all vectors (through $\bar\beta_{l,j}^i$) and connect all vectors to all hosts but $i$ (through $w_{kj}$, with $k\neq i$). By using Cayley's formula \cite{aigner1998cayley,moon1970counting}, each $v_i=C_{ii}$ is the sum of $n^{m-1}p^{m-1}{m}^{m-2}$ terms, each of which is the product of $m'-1$ $w_{kj}\bar\beta_{l,j}^i$ with $k\neq i$. These $w_{kj}\bar\beta_{l,j}^i$ represent the weight of each spanning tree $\mathbb T_i$, rooted at Host $i$. \\

Now, we investigate what each terms of $v_iw_{ij}\bar\beta_{l,j}^r$ represents in terms of the graph  $G(\bar B)$ and its spanning trees. Indeed, each term in $v_iw_{ij}\bar\beta_{l,j}^r$ for all $i,\; j,\; l,\; r$ is a weight of a unicyclic graph of a particular length, obtained by adding an arc $(r,l,i,j)$ to a directed tree rooted at Host $i$, $T\in\mathbb T_i$. The obtained unicyclic graph $Q$ has a unique cycle $CQ$ of length $1\leq d\leq m$. We group all isomorphic cycles as they will have the same coefficients. Hence, we conclude that $\mathcal S_1+\mathcal S_2$ is the sum over all unicyclic graphs as  
  $\mathcal S_1+\mathcal S_2= \sum_{Q}\mathcal S_{Q},
 $ 
 where
 \begin{multline}\label{eq:SumOnCycles}
\mathcal S_{Q}=w(Q_1)\sum_{(i,1,r,j)\in E(CQ_1)}\left( 5-\frac{S_i^\ast}{S_i} -\frac{E_i^\ast S_{i}I_{v,j}}{E_{i} S_{i}^\ast I_{v,j}^\ast}
-\frac{E_{i}}{E_{i}^\ast}\frac{I_{1,i}^\ast}{I_{1,i}}
-\frac{S_{v,j}^\ast}{S_{v,j}}
 -\frac{I_{v,j}^\ast S_{v,j}I_{1,r}}{I_{v,j} S_{v,j}^\ast I_{1,r}^\ast} \right) \\
  +w(Q_2)\sum_{(i,l,r,j)\in E(CQ_2)}\left( 4+l-\frac{S_i^\ast}{S_i} -\frac{E_i^\ast}{E_{i}}\frac{S_{i}}{S_{i}^\ast}\dfrac{I_{v,j}}{I_{v,j}^\ast}
-\frac{E_{i}}{E_{i}^\ast}\frac{I_{1,i}^\ast}{I_{1,i}}
-\frac{S_{v,j}^\ast}{S_{v,j}}
 -\frac{I_{v,j}^\ast}{I_{v,j}}   \frac{S_{v,j}}{S_{v,j}^\ast}\frac{I_{l,r}}{I_{l,r}^\ast}\right. \\
  \left.-\sum_{k=2}^l\frac{I_{k-1,r}}{I_{k-1,r}^\ast}\frac{I_{k,r}^\ast}{I_{k,r}}
 \right)\ \\
 +w(Q_3)\sum_{(i,l,r,j)\in E(CQ_3)}\left( 9+l-\frac{S_i^\ast}{S_i} -\frac{E_i^\ast}{E_{i}}\frac{S_{i}}{S_{i}^\ast}\dfrac{I_{v,j}}{I_{v,j}^\ast}
-\frac{E_{i}}{E_{i}^\ast}\frac{I_{1,i}^\ast}{I_{1,i}}
-\frac{S_{v,j}^\ast}{S_{v,j}}
 -  \frac{I_{v,j}^\ast S_{v,j}I_{1,r}}{I_{v,j} S_{v,j}^\ast I_{1,r}^\ast} \right. \\
 \left.-\frac{S_i^\ast}{S_i} -\frac{E_i^\ast}{E_{i}}\frac{S_{i}}{S_{i}^\ast}\dfrac{I_{v,j}}{I_{v,j}^\ast}
-\frac{E_{i}}{E_{i}^\ast}\frac{I_{1,i}^\ast}{I_{1,i}}
-\frac{S_{v,j}^\ast}{S_{v,j}}
 -\frac{I_{v,j}^\ast}{I_{v,j}}   \frac{S_{v,j}}{S_{v,j}^\ast}\frac{I_{l,r}}{I_{l,r}^\ast} 
-\sum_{k=2}^l\frac{I_{k-1,r}}{I_{k-1,r}^\ast}\frac{I_{k,r}^\ast}{I_{k,r}}
 \right),
 \end{multline}
where $CQ_1$, $CQ_2$ (for which $l\geq2$) and $CQ_3$ represents the cycles that correspond to elements in $\mathcal S_1$, $\mathcal S_2$ and $\mathcal S_1+\mathcal S_2$ exclusively. Now, each of the three sums in (\ref{eq:SumOnCycles}) are definite-negative. Indeed, we have:
$$
\prod_{(i,1,r,j)\in E(CQ_1)}\frac{S_i^\ast}{S_i} \frac{E_i^\ast S_{i}I_{v,j}}{E_{i} S_{i}^\ast I_{v,j}^\ast}
\frac{E_{i}}{E_{i}^\ast}\frac{I_{1,i}^\ast}{I_{1,i}}
\frac{S_{v,j}^\ast}{S_{v,j}}
 \frac{I_{v,j}^\ast S_{v,j}I_{1,r}}{I_{v,j} S_{v,j}^\ast I_{1,r}^\ast}
 =
 \prod_{(i,1,r,j)\in E(CQ_1)} 
\frac{I_{1,i}^\ast}{I_{1,i}}
 \frac{ I_{1,r}}{  I_{1,r}^\ast}=1,
$$
since $CQ_1$ is a cycle. Also,
\begin{multline*}\prod_{(i,l,r,j)\in E(CQ_2)}\frac{S_i^\ast}{S_i}\frac{E_i^\ast}{E_{i}}\frac{S_{i}}{S_{i}^\ast}\dfrac{I_{v,j}}{I_{v,j}^\ast}
\frac{E_{i}}{E_{i}^\ast}\frac{I_{1,i}^\ast}{I_{1,i}}
\frac{S_{v,j}^\ast}{S_{v,j}}
 \frac{I_{v,j}^\ast}{I_{v,j}}   \frac{S_{v,j}}{S_{v,j}^\ast}\frac{I_{l,r}}{I_{l,r}^\ast} 
\prod_{k=2}^l\frac{I_{k-1,r}}{I_{k-1,r}^\ast}\frac{I_{k,r}^\ast}{I_{k,r}}\\
=\prod_{(i,l,r,j)\in E(CQ_2)}  
 \frac{I_{1,i}^\ast}{I_{1,i}}
 \frac{I_{l,r}}{I_{l,r}^\ast} 
\prod_{k=2}^l\frac{I_{k-1,r}}{I_{k-1,r}^\ast}\frac{I_{k,r}^\ast}{I_{k,r}}=1.
\end{multline*}
Similarly, the product of the non constant terms in the last sum  (\ref{eq:SumOnCycles}) is equal to 1 as $CQ_3$ is a cycle.

Hence, by the arithmetic-geometric mean, $\mathcal S_{Q}$ is definite-negative for each unicyclic graph $Q$. It follows from (\ref{LyapDotHostGeneral5}), that $\dot{\mathcal V}$ is definite-negative. Thus, the global stability of the EE follows from the Lyapunov stability theorem. This ends the proof.
  \end{proof}
  \begin{remark}
 Although the Lyapunov function obtained in Theorem \ref{GASEE} is somehow similar to those obtained in \cite{Guo_li_CAMQ_06,Guo_li_PAMS08} for multi-group models, they are structurally different, in the sense that it is not linear combination of Lyapunov functions of one-group (or one host, multiple vectors). Indeed, the orbital derivative of Lyapunov function $\mathcal V=\sum_{i=1}^mv_i\mathcal V_i$ where is $\mathcal V_i$ is the Lyapunov function for one host and multiple vectors, is not definite-definite negative for the coefficient $v=(v_1,v_2,\dots,v_m)^T$.
 \end{remark}
   
    In this paper, we have considered an $SI$ model structure for the dynamics of the vectors. However, the results obtained here are valid for $SEI$ type of structure. Indeed, in this case, the following Lyapunov function
     $\mathcal V=\sum_{i=1}^mv_i\mathcal V_{i}$
  where
  \begin{eqnarray*}
  \mathcal V_{i}&=&\int_{S_i^\ast}^{S_i}\frac{x-S_i^\ast}{x}dx+\int_{E_{i}^\ast}^{E_i}\frac{x-E_i^\ast}{x}dx+\sum_{k=1}^nc_{k,i}\int_{I_{k,i}^\ast}^{I_{k,i}}\frac{x-I_{k,i}^\ast}{x}dx+\sum_{j=1}^pw_{ij}\mathcal V_{v,j}
  \end{eqnarray*}
where $$\mathcal V_{v,j}=\int_{S_{v,j}^\ast}^{S_{v,j}}\frac{x-S_{v,j}^\ast}{x}dx     +\int_{E_{v,j}^\ast}^{E_{v,j}}\frac{x-E_{v,j}^\ast}{x}dx +\frac{\nu_{v,j}+\mu_{v,j}+\delta_{v,j}}{\nu_{v,j}}\int_{I_{v,j}^\ast}^{I_{v,j}}\frac{x-I_{v,j}^\ast}{x}dx,$$ $$w_{ij}=\frac{a_{ij}\beta_{ij}^\diamond S_i^\ast\nu_{v,j}}{N_i(\nu_{v,j}+\mu_{v,j}+\delta_{v,j})(\mu_{v,j}+\delta_{v,j})},\quad c_{1,i}=\frac{\nu_i+\mu_i}{\nu_i}.$$
And
$$c_{k,i}=\frac{1}{\alpha_{k,i}I_{k,i}^\ast}\sum_{j=1}^pw_{ij}\sum_{i=1}^ma_{ij}\frac{S_{v,j}^\ast}{N_i}\sum_{l=k}^n\beta_{l,j}^iI_{l,i}^\ast\quad \forall k=1,\dots,n.$$
This Lyapunov function has its derivative along the trajectories of the $SEI^nR-SEI$ system, definite negative. The corresponding coefficients $v_i$ are determined in the same fashion as the $SEI^nR-SI$ case.

     \section{Conclusion and Discussions}\label{ConclusionDiscussions}
     In this paper, we formulated a multi-host, multi-stage and multi-vector epidemic model that describes the evolution of a class of zoonoses in which the pathogen is shared by multiple host species and the transmission occurs through the biting or landing of an arthropod vector. The proposed model improves those of  \cite{BicharaIggidrSmith2017,johnson2016modeling,palmer2018dynamics}  ---by incorporating multiple arthropod vector species--- and  \cite{cruz2012control,cruz2012multi}  ---by incorporating multiple hosts within multiple stages in each hosts' infectious class and the heterogeneous nature of the interactions between these host species and multiple arthropod vector species.
 We computed the basic reproduction number of the general system $\mathcal{R}_0^2(m,n,p)$, for $m$ hosts, $n$ hosts' infectious stages and $p$ vectors species. We proved that the disease free equilibrium is globally asymptotically stable whenever $\mathcal{R}_0^2(m,n,p)<1$ (\Cref{GASDFE}).  Under the assumption the host-vector network configuration is irreducible, we proved that there exists a unique ``strongly" endemic equilibrium as long as $\mathcal{R}_0^2(m,n,p)>1$ and that, it is GAS whenever it exists. This result is new and improves previous results of \cite{BicharaIggidrSmith2017}, for which the global result for multiple hosts and one vector is given, and of \cite{cruz2012multi} in which a case of multi-species, multi-vector is considered but no stability results were given.
 
 The global stability of the strongly equilibrium relies on a carefully constructed  Lyapunov functions and tools of graph theory, \`a la \cite{Guo_li_CAMQ_06,Guo_li_PAMS08,iggidr2014dynamics}. 
 The uniqueness and the global stability of the strongly endemic equilibrium requires the irreducibility of the host-vector network (\Cref{ExistenceUniquenessEE} and \Cref{GASEE}), leading to the conclusion that the disease either dies out or persists in all hosts and all vectors. That is, controlling the disease requires an intervention in all hosts. This is close to impossible, given that some hosts are not even known. Furthermore, in some cases the natural habitats of some hosts and vectors are so apart that it is unlikely there is a direct transmission between these hosts and vector for an infection to take place. This could collapse the irreducibility of the hosts-vectors configuration. Thus, it is important to investigate the global asymptotic behavior of the solutions when the host-vector network configuration is reducible. This is the subject of a separate study and will be published elsewhere.

 Other venues of expanding this work consist of considering different functional reproduction schemes for different vectors and/or hosts. Indeed, in this paper, although we have different reproduction (recruitment) rate for each hosts and vectors, they follow the same scheme, that is, they all consists of constant recruitments.

  \section*{Acknowledgments}
The author  would like to thank B. K. Druken, A. Iggidr and L. Smith for helpful discussions.


\end{document}